\documentclass[smallcondensed,runningheads,twosided]{svjour3}
\usepackage{amsmath,amssymb,mathtools,listings}
\usepackage{graphicx}
\journalname{Submitted}
\newcommand*{\eps}{{\varepsilon}}
\spnewtheorem{algorithm}{Algorithm}{\bf}{\rm}
\smartqed

\begin{document}

\title{Chi-square simulation of the CIR process and the Heston model}
\author{Simon J.A. Malham \and Anke Wiese}
\authorrunning{Malham and Wiese}

\institute{
Simon J.A. Malham \and Anke Wiese \at
Maxwell Institute for Mathematical Sciences \\
and School of Mathematical and Computer Sciences\\
Heriot-Watt University, Edinburgh EH14 4AS, UK \\
Tel.: +44-131-4513200\\
Fax: +44-131-4513249\\
\email{simonmalham@gmail.com}\\
\email{A.Wiese@hw.ac.uk}}

%\dedication{Dedicated to Mrs Marlies Wiese}

\date{Received: 2nd July 2012}

\voffset=10ex 

\maketitle
\begin{abstract}
The transition probability of a Cox--Ingersoll--Ross 
process can be represented by a non-central chi-square density. 
First we prove a new representation for the central chi-square density 
based on sums of powers of generalized Gaussian random variables.
Second we prove Marsaglia's polar method extends to this distribution,
providing a simple,  exact, robust and efficient 
acceptance-rejection method for generalized Gaussian sampling and 
thus central chi-square sampling. 
Third we derive a simple, high-accuracy, robust and efficient 
direct inversion method for generalized Gaussian sampling based
on the Beasley--Springer--Moro method. 
Indeed the accuracy of the approximation to the inverse 
cumulative distribution function is to the tenth decimal place. 
We then apply our methods to non-central chi-square variance sampling
in the Heston model. We focus on the case when the number of 
degrees of freedom is small and the zero boundary is attracting and 
attainable, typical in foreign exchange markets.  Using the additivity
property of the chi-square distribution, our methods apply in
all parameter regimes. 
\end{abstract}

\section{Introduction}\label{sec:Intro}
The mean-reverting square-root process or  Cox--Ingersoll--Ross (CIR) process is frequently used in
finance and economics to model the evolution of key financial variables, most notably 
%the Cox--Ingersoll--Ross model for interest rates 
to model the short rate of interest (Cox, Ingersoll and Ross~\cite{CIR}) and
in the Heston stochastic volatility model (Heston~\cite{H}). 
% for the stochastic variance of the log returns of
%of a financial variable such as a stock index. 
Other applications include the modelling of mortality intensities
and of default intensities in credit risk models,  for example. %among others 
The CIR process can be expressed in the form
\begin{equation*}
\mathrm{d}V_t=\kappa(\theta-V_t)\,\mathrm{d}t
+\varepsilon\sqrt{V_t}\,\mathrm{d}W^1_t,
\end{equation*} 
where $W^1$ is a Wiener process and  $\kappa$, $\theta$
and $\varepsilon$ are positive constants. 
The Heston model is a two-factor model,
in which one component $S$ describes the evolution of a financial variable
such as a stock index or exchange rate,
and another component $V$ describes the stochastic variance of its returns.
Indeed the stochastic variance $V$ evolves according to the CIR process
described above.  The Heston model is then completed by prescribing
the evolution of $S$ by 
\begin{equation*}
\mathrm{d}S_t=\mu S_t\,\mathrm{d}t
+\sqrt{V_t}\,S_t\,\bigl(\rho\,\mathrm{d}W^1_t
+\sqrt{1-\rho^2}\,\mathrm{d}W^2_t\bigr),
%\mathrm{d}V_t=&\;\kappa(\theta-V_t)\,\mathrm{d}t
%+\varepsilon\sqrt{V_t}\,\mathrm{d}W^1_t,
\end{equation*} 
where $W^2_t$ is an independent scalar Wiener process. 
The additional parameter $\mu$ is positive, while $\rho\in(-1,1)$.
 
The transition probability density of the CIR process is
known explicitly, it can be represented by a non-central chi-square density.  
Depending on the number of degrees of freedom 
$\nu\coloneqq4\kappa\theta/\varepsilon^2$,
there are fundamental differences in the
behaviour of the CIR process. If $\nu$ is larger
or equal to 2, the zero boundary is unattainable; if it is smaller than 2, 
the zero boundary is attracting and attainable. At the zero boundary though,
the solution is immediately reflected into the positive domain. This behaviour in
the latter case is particularly difficult to capture numerically.

A number of successful simulation schemes have been
developed for the non-attainable zero boundary case. 
There are schemes based on implicit time-stepping integrators, 
see for example Alfonsi~\cite{Al2005},  Kahl and Schurz~\cite{KS}
and Dereich, Neuenkirch and Szpruch~\cite{DNS}. 
Other time discretization approaches
involve splitting the drift and diffusion vector fields
and evaluating their separate flows (sometimes exactly)
before they are recomposed together, typically using the Strang ansatz.
See for example Higham, Mao and Stuart~\cite{HMS} and
Ninomiya and Victoir~\cite{NV}. 
However, these splitting methods and the implicit methods 
only apply in the non-attracting zero boundary case.  
Recently, Alfonsi~\cite{Al2007} has combined a splitting method with 
an approximation using a binary random variable near the zero boundary 
to obtain a weak approximation method for the full parameter regime.
Moro and Schurz~\cite{MS} have also successfully combined 
exponential splitting with exact simulation. Dyrting~\cite{Dy}
outlines and compares several different series and asymptotic
approximations for non-central chi-square distribution.

Other direct discretization approaches, that can be applied to
the attainable and unattainable zero boundary case
are based on forced Euler-Maruyama approximations and typically 
involve negativity truncations; some of these methods are 
positivity preserving. See for example Deelstra and Delbaen~\cite{DD},
Bossy and Diop~\cite{BD} and also Berkaoui, Bossy and Diop~\cite{BBD}, 
Lord, Koekkoek and Van Dijk~\cite{LKVD},
as well as Higham and Mao~\cite{HM}, among others. 
These methods  all converge to the exact solution, 
but their rate of strong convergence and discretization errors are 
difficult to establish. The full truncation method of Lord, Koekkoek and 
Van Dijk~\cite{LKVD} has in practice shown to be the leading method 
in this class.
%, but can
%produce significant discretization biases in some practical applications, 
%see for example Andersen~\cite{An} and Haastrecht and Pelsser~\cite{HP}.

Exact simulation methods typically sample from the known 
non-central chi-square distribution $\chi_\nu^2(\lambda)$ for the transition
probability of the CIR process $V$ (see Cox, Ingersoll and Ross~\cite{CIR} and 
Glasserman~\cite[Section~3.4]{Glass}). Broadie and Kaya~\cite{BK}
proposed sampling from $\chi_\nu^2(\lambda)$ as follows. When $\nu>1$, 
$\chi_\nu^2(\lambda)=\bigl(\text{N}(0,\sqrt{\lambda})\bigr)^2+\chi_{\nu-1}^2$,
so such a sample can be generated by a standard Normal sample
and a central chi-square sample. When $0<\nu<1$, such a sample
can be generated by sampling from a Poisson distribution with mean $\lambda/2$,
and then sampling from a central $\chi_{2N+\nu}^2$ distribution.

In the Heston model, 
to simulate the asset price Broadie and Kaya integrated the 
variance process $V$ to obtain an expression for $\int\sqrt{V_\tau}\,\mathrm{d}W_\tau$.
They substituted that expression into the stochastic differential equation 
for $\ln S_t$. The most difficult task left is then to simulate 
$\int V_\tau\,\mathrm{d}\tau$ on the global interval of integration conditioned
on the endpoint values of $V$; see Smith~\cite{Smith}. 
The Laplace transform of the transition density for this integral is
known from results in Pitman and Yor~\cite{PY}.
Broadie and Kaya used Fourier inversion techniques to sample from
this transition density. Glasserman and Kim~\cite{GK} on the other
hand, showed that linear combinations of series of particular 
gamma random variables exactly sample this density. They used
truncations of those series to generate suitably accurate 
sample approximations. 
Their method has proved to be highly effective in applications that do
not require the simulation of intermediate values of the process $S$,
for example when pricing derivatives that are not path-dependent.
Anderson~\cite{An} suggested two approximations
that make simulation of the Heston model very efficient, 
 and allow for pricing path-dependent options. 
The first was, after discretizing the 
time interval of integration for the price process, to approximate 
$\int V_\tau\,\mathrm{d}\tau$ on the integration subinterval by a
trapezoidal rule. %simple quadrature. 
This would thus require non-central $\chi_\nu^2(\lambda)$
samples for the volatility at each timestep. The second 
was to approximate and thus efficiently sample the $\chi_\nu^2(\lambda)$ 
distribution---in two different ways depending on the size of $\lambda$. 
Haastrecht and Pelsser~\cite{HP} have recently introduced a rival 
$\chi_\nu^2(\lambda)$ sampling method to Andersen's which
utilizes, for small $\lambda$, pre-caching tables for central chi-square
$\chi^2_{\nu+2N}$ distributions for small values of $N$, and for
large~$\lambda$, a matched squared normal random variable.

There are also numerous approximation methods based on the corresponding
Kolmogorov or Fokker--Planck partial differential equation. 
These can take the form of Fourier transform methods---see Carr and Madan~\cite{CM},
Kahl and J\"ackel~\cite{KJ2} or Fang and Oosterlee~\cite{FO1,FO2} 
for example---or some involve direct discretization of the 
Kolmogorov equation---see in 't Hout and Foulon~\cite{Karel} and Haentjens  
and in 't Hout~\cite{HHW_karel}.

We focus on the challenge of the attainable zero boundary case and in particular
on the case when $\nu\ll\,1$, typical of FX markets and long-dated
interest rate markets as remarked in Andersen~\cite{An}, and also
observed in credit risk, see  Brigo and Chourdakis~\cite{BC}.
(Using the additivity property of the chi-square distribution
$\chi^2_{\nu+k} = \chi^2_\nu +\chi^2_k$, the results can be
straightforwardly extended to all parameter regimes.)
The method we propose follows the lead of Andersen~\cite{An}, we
approximate the integrated variance process $\int V_\tau\,\mathrm{d}\tau$
by a trapezoidal rule. For the simulation of
the non-central $\chi_\nu^2(\lambda)$ transition density of the CIR process
that is used to model the variance process
required for each timestep of this integration method,
we suggest two new methods that rely on the following representation. 
A non-central $\chi_\nu^2(\lambda)$ random variable can be generated 
from a central $\chi_{2N+\nu}^2$ random variable with $N$ 
chosen from a Poisson distribution with mean $\lambda/2$.
Further, a $\chi_{2N+\nu}^2$ random variable can be generated 
from the sum of squares of $2N$ independent standard Normal random variables
(more efficiently sampled as the sum of the logarithm of $N$ uniform random variables)
and an independent central $\chi_{\nu}^2$ random variable. So the question
we now face is how can we efficiently simulate a central $\chi_{\nu}^2$ random variable,
especially for $\nu<1$? Suppose that $\nu$ is rational and expressed
in the form $\nu=p/q$ with $p$ and $q$ natural numbers. 
We show that a central $\chi_{\nu}^2$ random variable
can be generated from the sum of the $2q$th power of $p$  
independent random variables chosen from a generalized
Gaussian distribution $\text{N}(0,1,2q)$.

The question now becomes, how can we sample from a 
$\text{N}(0,1,2q)$ distribution? We have two answers.
The first lies in generalizing Marsaglia's polar method
for pairs of independent standard Normal random variables, which
we call the Marsaglia generalized Gaussian method (MAGG).
%Indeed we generate $2q$ uniform random variables $U=(U_1,\ldots,U_{2q})$ 
%over $[-1,1]$, and condition on their $2q$th norm $\|U\|_{2q}$,
%being less than unity. Then we prove that the $2q$ random variables 
%$U\cdot (-2\log\|U\|_{2q}^{2q})^{1/2q}/\|U\|_{2q}$ are independent 
%$\text{N}(0,1,2q)$ random variables. 
The second method generalizes  the Beasley--Springer--Moro direct inversion
method for standard Normal random variables to generate 
a high accuracy approximation
of the inverse generalized Gaussian distribution function.
We provide a thorough comparison, of our generalized Marsaglia polar
method and of our direct inversion method for 
sampling from the central $\chi_{\nu}^2$ distribution, to the 
acceptance-rejection methods of Ahrens--Dieter and 
Marsaglia--Tsang (see Ahrens and Dieter~\cite{AD:chi}, 
Glasserman~\cite{Glass} and Marsaglia and Tsang~\cite{MarTsa}).

The CIR process can thus be simulated by the two approaches just described;
exactly in the first instance and with very high accuracy in the second. 
The advantages of both  approaches are that for
the mean-reverting variance process in the Heston model, 
we can efficiently generate high quality samples simply and robustly.
The  methods require the degrees of freedom to be rational,
however this is fulfilled in practical applications: the parameter
$\nu$ will typically be obtained through calibration and can only be 
computed up to a pre-specified accuracy.
We demonstrate our two methods in the computation 
of option prices for parameter cases that are considered 
in Andersen~\cite{An} and Glasserman and Kim~\cite{GK} 
and described there as challenging and practically relevant. We also
demonstrate our methods for the pricing of path-dependent derivatives.

To summarize, we:
\begin{itemize}
\item Prove that a central chi-squared random variable with less than one degree
of freedom, can be written as a sum of powers of generalized Gaussian
random variables;
\item Prove a new method---the generalized Marsaglia polar method---for 
generating generalized Gaussian samples;
\item Provide a new and fast high-accuracy approximation 
(in principle to machine error) %for values up to $1-10^{-8}$ on the abscissa 
to the inverse generalized Gaussian distribution function;
\item Establish two new simple, flexible, high-accuracy,  efficient methods for
%exact, unbiased and efficient methods for 
simulating the Cox--Ingersoll--Ross process,  for an attracting and 
attainable zero boundary, which we apply to simulating the Heston model.
%and thus establish two
%new simple methods for simulating the Heston model.
\end{itemize}

Our paper is organised as follows. In Section~\ref{sec:ggsampling} we present
our new generalized Marsaglia method and our direct inversion method
for sampling from the generalized Gaussian distribution. 
In Section~\ref{sec:CChisq} we derive the representation of a 
chi-square distributed random variable as a sum of powers 
of independent generalized Gaussian random variables. 
We include a thorough comparison of the generalized Marsaglia method
and the direct inversion method for the central chi-squared distribution
(based on sampling from the generalized Gaussian distribution)
with the acceptance-rejection methods of Ahrens and Dieter~\cite{AD:chi} 
and of Marsaglia and Tsang~\cite{MarTsa}. We apply both our  methods to the 
CIR process and Heston model in Section~\ref{sec:simulations}. 
We compare their accuracy and efficiency to the leading 
approximation method of Andersen~\cite{An}.
Finally in Section~\ref{sec:conclu} we present some concluding remarks.

\section{Generalized Gaussian sampling}\label{sec:ggsampling}
We require an efficient method for generating generalized Gaussian samples.
Here we provide two such methods. The first method is a generalization of 
Marsaglia's polar method for standard Normal random variables.
This is an exact acceptance-rejection method. The second
method is a direct inversion method that generalizes the 
Beasley--Springer--Moro method for standard Normal random variables. 
In principle this method is accurate to machine error. 

\begin{definition}[Generalized Gaussian distribution]
A generalized $\mathrm{N}(0,1,q)$ random variable,
for $q\geqslant1$, has distribution function for $x\in\mathbb R$: 
\begin{equation*}
\Phi(x)\coloneqq\gamma_q
\int_{-\infty}^x\exp\bigl(-|\tau|^{q}/2\bigr)\,\mathrm{d}\tau,
\end{equation*}
where $\gamma_q\coloneqq q/\bigl(2^{1/q+1}\Gamma(1/q)\bigr)$ 
and $\Gamma(\cdot)$ is the standard gamma function.
\end{definition}
See Gupta and Song~\cite{GS}, Song and Gupta~\cite{SG}, Sinz,
Gerwinn and Bethge~\cite{SGB}, Sinz and Bethge~\cite{SB} 
and Pog\'any and Nadarajah~\cite{PN}
for more details on this distribution and its properties.

\subsection{Generalized Marsaglia polar method}
We generalize Marsaglia's polar method for pairs of 
independent standard Normal random variables
(see Marsaglia~\cite{Mar}).
\begin{theorem}[Generalized Marsaglia polar method]
Suppose for some $q\in\mathbb N$ that $U_1,\ldots,U_{q}$
are independent identically distributed 
uniform random variables over $[-1,1]$. 
Condition this sample set to satisfy the requirement $\|U\|_{q}<1$,
where $\|U\|_{q}$ is the $q$-norm of $U=(U_1,\ldots,U_{q})$.
Then the $q$ random variables generated by
$U\cdot (-2\log\|U\|_{q}^{q})^{1/q}/\|U\|_{q}$ are independent 
$\text{N}(0,1,q)$ distributed random variables. 
\end{theorem}
\begin{proof}
Suppose for some $q\in\mathbb N$ that $U=(U_1,\ldots,U_{q})$ 
are independent identically distributed 
uniform random variables over $[-1,1]$, conditioned on the requirement
that $\|U\|_{q}<1$. Then the scalar variable
$Z\coloneqq(-2\log\|U\|_q^q)^{1/q}>0$
is well defined. Let $f$ denote the probability density
function of $U$ given $\|U\|_{q}<1$; it is defined on the interior of 
the $q$-sphere, $\mathbb S_q(1)$, whose bounding surface is $\|U\|_{q}=1$.
We define a new set of $q$ random variables $W=(W_1,\ldots,W_q)$ 
by the map $G\colon\mathbb S_q(1)\to\mathbb R^p$ where $G\colon U\mapsto W$
is given by $G\circ U=(Z/\|U\|_q)\cdot U$.
Note that the inverse map $G^{-1}\colon\mathbb R^p\to\mathbb S_q(1)$ is
well defined and given by $G^{-1}\circ W=Z^{-1}\cdot\exp(-Z^q/2q)\cdot W$,
where we note that in fact $Z=\|W\|_q$ which comes from taking
the $q$-norm on each side of the relation $W=G(U)$.
We wish to determine the probability density function of $W$. Note that
if $\Omega\subset\mathbb R^q$, then we have
$\mathbb P\,\bigl(W\in\Omega\bigr)
=\mathbb P\,\bigl(U\in G^{-1}(\Omega)\bigr)
=\int_{G^{-1}(\Omega)}f\circ u\,\mathrm{d}u
=\int_{\Omega}(f\circ G^{-1}\circ w)\cdot\bigl|\det(\mathrm{D}G^{-1}\circ w)\bigr|\,\mathrm{d}w$,
where for $w=(w_1,\ldots,w_p)\in\Omega$, the quantity $\mathrm{D}G^{-1}\circ w$
denotes the Jacobian transformation matrix of $G^{-1}$. Hence the 
probability density function of $W$ is given by 
$(f\circ G^{-1}\circ w)\cdot\bigl|\det(\mathrm{D}G^{-1}\circ w)\bigr|$.
The Jacobian matrix and its determinant are established by direct
computation. For each $i,k=1,\ldots,q$ we see that if we define
$g(z)\coloneqq-(1/2+1/z^q)$, then 
\begin{equation*}
\partial G^{-1}_k/\partial w_i
=z^{-1}\exp(-z^q/2q)\cdot\Bigl(\delta_{ik}
+g(z)\cdot\bigl(\mathrm{sgn}(w_i)\cdot|w_i|^{q-1}\bigr)\cdot w_k\Bigr),
\end{equation*}
where $\delta_{ik}$ is the Kronecker delta function. If we set
$v=\bigl(\mathrm{sgn}(w_1)\cdot|w_1|^{q-1},\ldots,\mathrm{sgn}(w_q)\cdot|w_q|^{q-1}\bigr)^{\mathrm{T}}$
then we see that our last expression generates the following relation for
the Jacobian matrix: 
$z\exp(z^q/2q)\cdot\bigl(\mathrm{D}G^{-1}\circ w\bigr)
=I_q+g(z)\cdot v\,w^{\mathrm{T}}$,
where $I_q$ denotes the $q\times q$ identity matrix.
From the determinant rule for rank-one updates---see 
Meyer~\cite[p.~475]{Meyer}---we see that the determinant of the 
Jacobian matrix is given by
\begin{align*}
\det(\mathrm{D}G^{-1}\circ w\bigr)
&=z^{-q}\exp(-z^q/2)\cdot\bigl(1+g(z)\,w^{\mathrm{T}}\,v\bigr)\\
&=z^{-q}\exp(-z^q/2)\cdot\bigl(1+g(z)\,z^q\bigr)\\
&=-\tfrac12\exp(-z^q/2).
\end{align*}
Noting that
$\mathrm{vol}\bigl(\mathbb S_q(1)\bigr)=2^q\cdot\bigl(\Gamma(1/q)\bigr)^q/q^q$ 
we have
\begin{equation*}
(f\circ G^{-1}\circ w)\cdot\bigl|\det(\mathrm{D}G^{-1}\circ w)\bigr|
=\frac{q^q}{2^{q+1}\bigl(\Gamma(1/q)\bigr)^q}\cdot\exp(-z^q/2).
\end{equation*}
This is the joint probability density function for $q$ independent
identically distributed $q$-generalized Gaussian random variables,
establishing the required result. 
\qed
\end{proof}
\begin{remark} 
The corresponding generalization of the  Box--Muller method 
%for sampling standard Normal random variables
involves the beta distribution, which does not appear to be 
a convenient approach; see Liang and Ng~\cite{LNg}, 
Harman and Lacko~\cite{HL} and Lacko and Harman~\cite{LH}.
\end{remark}

\subsection{Direct inversion method}\label{subsec:di}
We generalize the Beasley--Springer--Moro direct inversion 
method for standard Normal random variables to the 
generalized Gaussian $\text{N}(0,1,q)$ for $q>2$; 
see Moro~\cite{Moro} or Joy, Boyle and Tan~\cite{JBT} for
the case $q=2$. 
We focus on the case of large $q$; the reasons for this 
will become apparent in Section~\ref{sec:CChisq}. In the 
limit of large $q$ the probability density function
for the generalized Gaussian attains the profile of
the density of a $\text{U}(-1,1)$ uniform random variable. 
For large but finite $q$ the probability density function
of the generalized Gaussian resembles a smoothed version
of the $\text{U}(-1,1)$ density profile. It naturally exhibits 
three distinct behavioural regimes which are also naturally
reflected in the generalized Gaussian distribution function
$\Phi$, as well as its inverse $\Phi^{-1}$ which is our
principal object of interest.
We use the symmetry of the density function
to focus on the positive half $[0,\infty)$ of its support. 
Correspondingly the inverse distribution
function $\Phi^{-1}$ is anti-symmetric about $1/2$;
and we can focus on the subinterval $[1/2,1)$ of
its support. Since $\Phi$ is monotonic (and bijective)
we naturally identify (and pairwise) label 
the three behavioural regimes of the density function
and inverse distribution function as follows. We set
$x_*\coloneqq \bigl(2(1-1/q)\bigr)^{1/q}$ and 
\begin{equation*}
x_\pm\coloneqq \Bigl(3(1-1/q)\pm\bigl((5-1/q)(1-1/q)\bigr)^{1/2}\Bigr)^{1/q}
\end{equation*}
and correspondingly $\Phi_*\coloneqq \Phi(x_*)$ and $\Phi_\pm\coloneqq \Phi(x_\pm)$.
Here $x_*\in[0,\infty)$ denotes the inflection point of
the density profile, i.e. $\Phi'''(x_*)=0$, while 
$x_\pm\in[0,\infty)$ are the points where $\Phi''''(x_\pm)=0$.
Then we identify the:
\begin{enumerate}
\item \emph{Central region} where $x\in[0,x_-]$ or
$\Phi\in[1/2,\Phi_-]$---roughly corresponding to the 
region where the density profile is flat and approximately
equal to $1/2$;
\item \emph{Middle region} where $x\in[x_-,x_+]$ or
$\Phi\in[\Phi_-,\Phi_+]$---roughly corresponding to the 
region where the density profile has a large negative slope; and
\item \emph{Tail region} where $x\in[x_+,\infty)$ or
$\Phi\in[\Phi_+,1)$---roughly corresponding to the 
region where the density profile is flat and approximately
equal to zero.
\end{enumerate}
%
%\begin{remark} 
%For $q=2$ there are only two distinct regions as
%the support of the middle region shrinks to zero;
%again see Moro~\cite{Moro} or Joy, Boyle and Tan~\cite{JBT}.
%\end{remark}
%
As in Beasley and Springer~\cite{BS}, in the \emph{central region} we 
approximate the inverse generalized Gaussian $\Phi^{-1}=\Phi^{-1}(u)$ 
with $u\in[1/2,\Phi_-]$ by an $(m,n)$ Pad\'e approximant
\begin{equation*}
\Phi^{-1}(u)\approx U
\cdot\frac{a_0+a_1(U^q)+\cdots+a_m(U^q)^m}{1+b_1(U^q)+\cdots+b_n(U^q)^n}, 
\end{equation*}
where $U=(u-1/2)/\gamma_q$, with $\gamma_q$
the reciprocal of the normalizing factor of the generalized
Gaussian distribution, and $a_0,\ldots,a_m,b_1,\ldots,b_n$ are
constant coefficients.
Typically across a large range of values of $q$
the choice of values for $m$ and $n$ equal to $3,4$ or $5$
generate approximations with order
of $10^{-10}$ accuracy. The coefficients $a_0,\ldots,a_m$ 
and $b_1,\ldots,b_n$ change as $q$ varies---see the discussion below. 

Motivated by the approximations suggested in  Blair, Edwards 
and Johnson~\cite{BEJ}, in the \emph{middle region} we approximate $\Phi^{-1}$
with $u\in[\Phi_-,\Phi_+]$ by a rational $(m,n)$ Pad\'e approximation 
of a scaled and shifted variable as follows: 
\begin{equation*}
\Phi^{-1}(u)\approx
\frac{c_0+c_1(\eta-\eta_*)+\cdots+c_m(\eta-\eta_*)^m}
{1+d_1(\eta-\eta_*)+\cdots+d_n(\eta-\eta_*)^n}, 
\end{equation*}
where $\eta\coloneqq-\log(1-u)$, $\eta_*\coloneqq-\log(1-\Phi_*)$
and $c_0,\ldots,c_m,d_1,\ldots,d_n$ are constant coefficients.
Note that the integers $m$ and $n$ here are distinct from those in
the central approximation above.
Again typically as $q$ varies, values of $m$ and $n$ equal to $3,4$ or $5$
generate approximations with order of $10^{-10}$ accuracy.

Now using the ansatz of Moro~\cite{Moro}, 
in the \emph{tail region} we approximate $\Phi^{-1}$
with $u\in[\Phi_+,1)$ by a degree $n$ Chebychev polynomial
approximation---suggested in Joy, Boyle and Tan~\cite{JBT}---of 
a scaled and shifted variable as follows: 
\begin{equation*}
\Phi^{-1}(u)\approx\hat c_0T_0(z)+\hat c_1T_1(z)+\cdots+\hat c_nT_n(z)-\hat c_0/2
\end{equation*}
where $T_n$ is the degree $n$ Chebychev polynomial 
in $z\in[-1,1]$, where $z\coloneqq k_1\xi+k_2$ and 
$\xi\coloneqq\log\bigl(-\log\bigl((1-u)/C_q\bigr)\bigr)$,
with $C_q\coloneqq1/\bigl(2\Gamma(1/q)\bigr)$.
The parameters $k_1$ and $k_2$ are chosen so that $z=-1$ when
$u=\Phi_+$ and $z=1$ when $u=1-10^{-12}$.
Then as $q$ varies, values of $n$ equal to $10$ 
generate approximants with order $10^{-10}$ accuracy.
As we discuss below, when we evaluate the Chebychev polynomials 
using double precision arithmetic, we need to restrict the tail approximation 
to $u\in[\Phi_+,1-10^{-8}]$.
\begin{remark} 
The choices of the scaled variables in the central and tail
region approximations above are motivated by the asymptotic
approximation for $\Phi^{-1}=\Phi^{-1}(u)$ as $u\to1^-$. After
applying the logarithm to the large $x$ asymptotic expansion 
for the generalized Gaussian distribution function $\Phi=\Phi(x)$,
we get for $y\coloneqq x^q/2$:
\begin{equation*}
y=-\log\biggl(\frac{1-\Phi(x)}{C_q}\biggr)+(1/q-1)\log y-y
+\log\biggl(1+\sum_{n\geqslant1}\frac{(1/q-1)\cdots(1/q-n)}{y^{n}}\biggr). 
\end{equation*}
We can generate an asymptotic expansion for $y$ and thus $\Phi^{-1}=(2y)^{1/q}$
by iteratively solving the above equation with initial guess given by 
the first term on the right shown above. The expansion is
$y\sim\mathrm{e}^\xi+(1/q-1)\xi+\mathrm{e}^{-\xi}P_1(\xi)+\mathrm{e}^{-2\xi}P_2(\xi)+\cdots$,
where the $P_n$ for $n\geqslant1$ are explicitly determinable degree $n$ polynomials.
\end{remark}
For some fixed values for $q$, we quote in Appendix~\ref{app:coeffs}
values for the coefficients for the approximations above in the 
three regions. We obtained the coefficients in the case of the two
Pad\'e rational approximants by applying the least squares approach 
advocated half way down page~205 in Press \emph{et al.}~\cite{PTVF}.
This requires values for $\Phi^{-1}$ at, for example, nodal points 
roughly distributed as zeros of a high degree Chebychev polynomial. 
These were obtained by high precision Gauss-Konrod quadrature 
approximation of $\Phi$ combined with a high precision root finding 
algorithm.
In the case of the Chebychev approximations in the tail region,
we computed the coefficients in the standard way, see for example
Section~5.8 in Press \emph{et al.}~\cite{PTVF}. Further Chebychev 
approximations can be efficiently evaluated using Clenshaw's
recurrence formula found on page~193 of Press \emph{et al.}.
Thus, with the $a$, $b$, $c$, $d$ and $\hat c$  coefficients 
above computed, our direct inversion algorithm is given in 
Appendix~\ref{app:di}.

Figure~\ref{fig:iggderr} shows the error in the inverse
generalized Gaussian approximations for $q=10,100,1000$.
The top three panels show that in all three cases, and
across the central, middle and tail regions, the coefficients
listed in Appendix~\ref{app:coeffs} generate approximations
with errors of order $10^{-10}$. This is comparable with
the error in the Beasley--Springer--Moro approximation
when $q=2$.  We note the tail region we have
considered extends to $1-10^{-8}$, whereas the tail region of
Beasley--Springer--Moro extends to $1-10^{-12}$. The reason for this
lies in the arithmetic precision we used. In fact we evaluated the 
coefficients for the Chebychev approximations in the tail regions using 
$25$ digit arithmetic, but evaluated the corresponding 
Chebychev polynomials in double precision arithmetic.
Indeed in the top three panels in Figure~\ref{fig:iggderr},
we can see the effect in tail regions of restricting 
calculations to double precision (16 digit) arithmetic. 
If we also evaluate the Chebychev polynomials in $25$ digit
arithmetic then errors of our Chebychev polynomial approximations
in the tail region are shown in the lower two panels in 
Figure~\ref{fig:iggderr}. We observe there that our tail
approximations in fact maintain $10^{-10}$ or better 
as far as $1-10^{-12}$ on the abscissa. However unless
stated otherwise, all our subsequent calculations
are performed using double precision arithmetic.

\begin{figure} 
  \begin{center}
  \includegraphics[width=10.0cm,height=3.3cm]{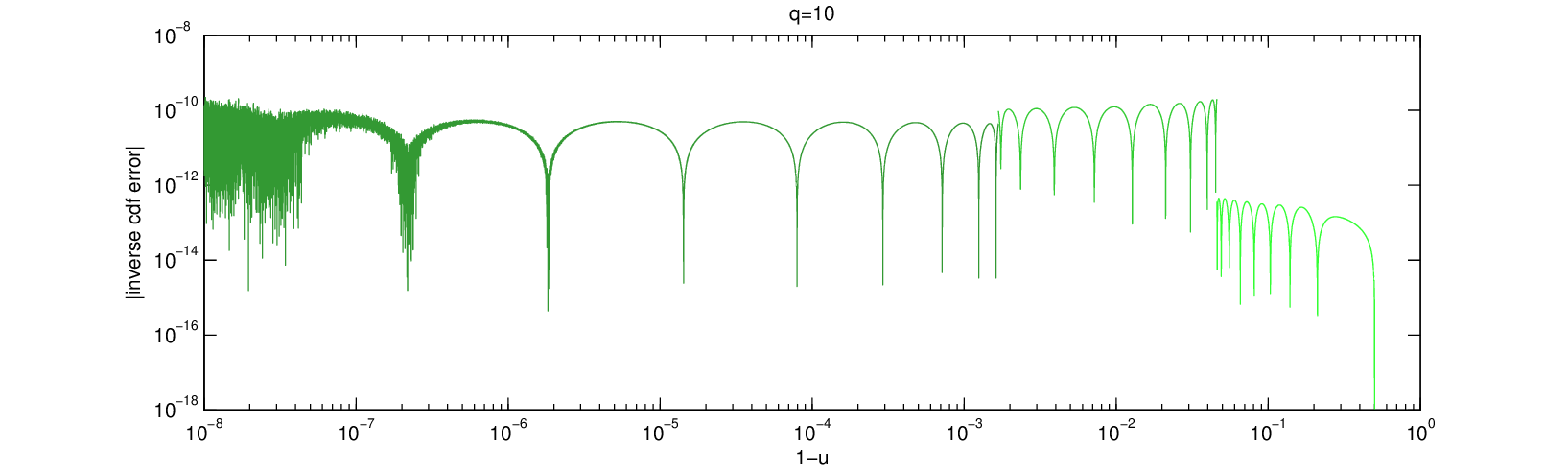}
  \includegraphics[width=10.0cm,height=3.3cm]{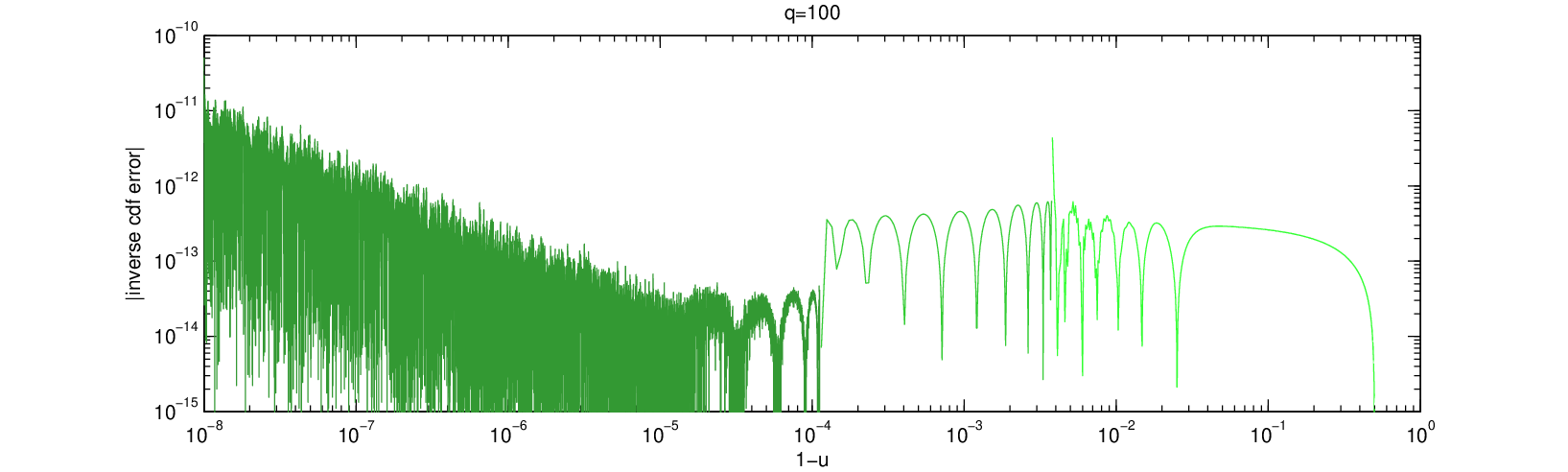} 
  \includegraphics[width=10.0cm,height=3.3cm]{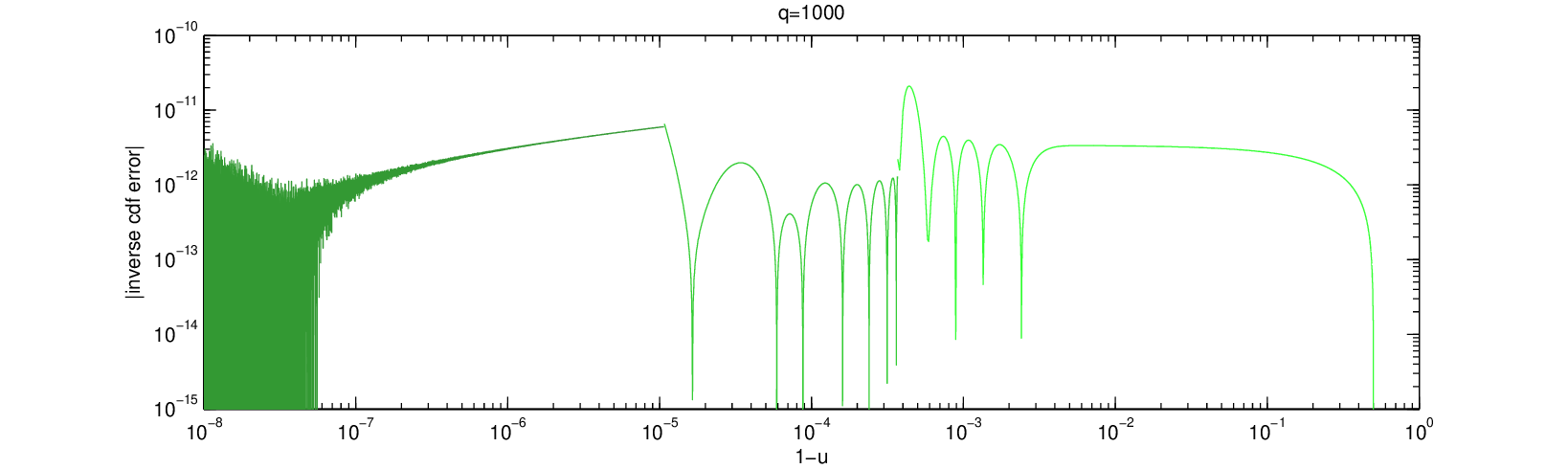}
  \includegraphics[width=10.0cm,height=3.3cm]{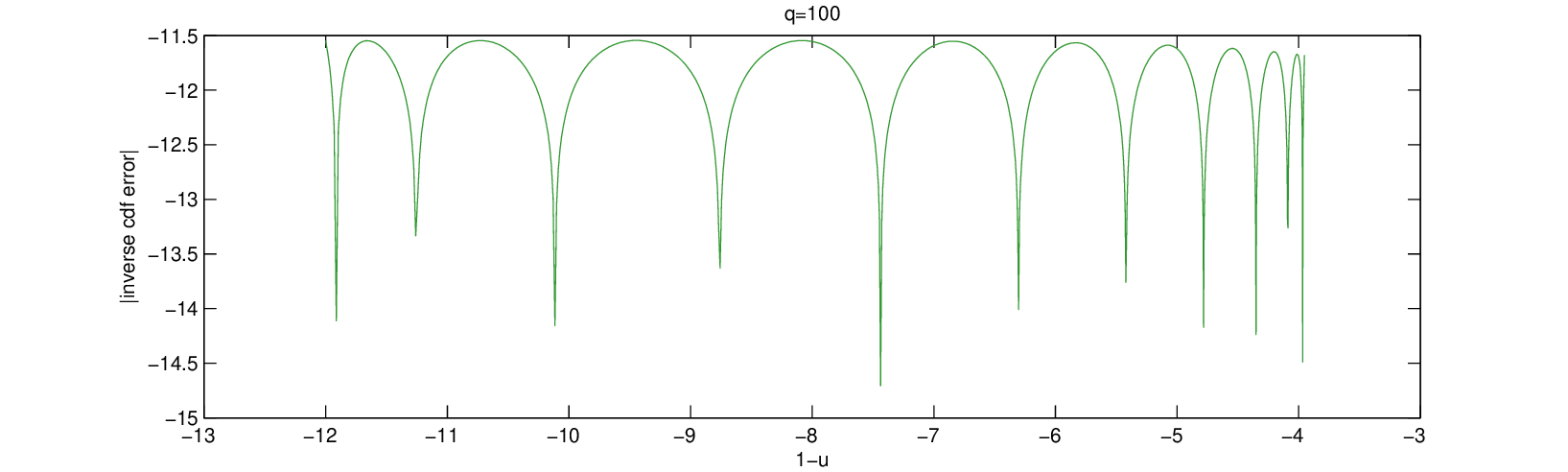} 
  \includegraphics[width=10.0cm,height=3.3cm]{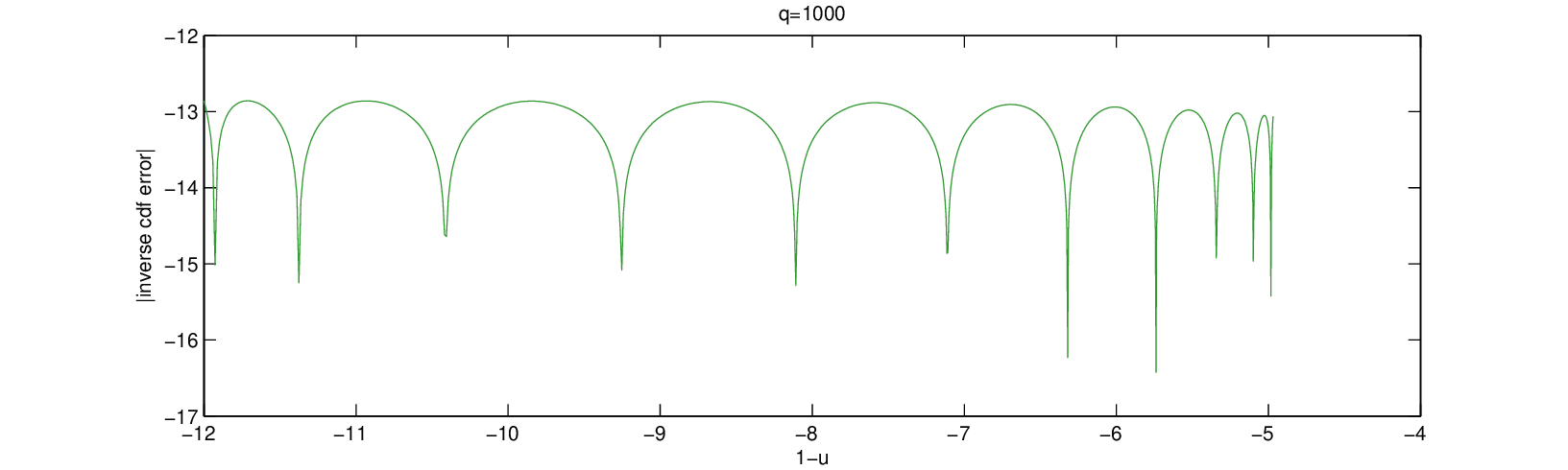}
  \end{center}
\caption{The top three panels show the error in our inverse
generalized Gaussian distribution approximations, respectively 
for $q=10,100,1000$, computed using double precision arithmetic. 
Each panel shows the central, middle
and tail approximants for the corresponding values of $q$.
Note the error in all cases is of order $10^{-10}$. 
The lower two panels show the error in our tail
approximations, respectively for $q=100,1000$, computed
using 25 digit arithmetic.}
\label{fig:iggderr}
\end{figure}

%\begin{remark} 
%In principle could use an algebraic computing package Mathematica 
%or Maple to derive approximations above with coeffs depending on $q$
%as Shaw~\cite{Shaw:stT} does for the Student T distribution.
%\end{remark}

\section{Chi-square sampling}\label{sec:CChisq}
We begin by proving that random variables with a central $\chi_{\nu}^2$ distribution, 
especially for $\nu<1$, can be represented by random variables with a generalized
Gaussian distribution. 
\begin{theorem}[Central chi-square from generalized Gaussians]
Suppose $X_i\sim\mathrm{N}(0,1,2q)$ are independent identically 
distributed random variables for $i=1,\ldots,p$, where $q\geqslant1$
and $p\in\mathbb N$. Then we have 
\begin{equation*}
\sum_{i=1}^p\bigl|X_i\bigr|^{2q}\sim\chi_{p/q}^2.
\end{equation*}
\end{theorem}
\begin{proof}
If $X\sim\text{N}(0,1,2q)$, then we have
$\mathbb P\bigl(|X|^{2q}<x\bigr)=2\cdot\mathbb P\bigl(0<X<|x|^{1/2q}\bigr)$
and a simple $2q$th power law transformation reveals that
$|X|^{2q}\sim\chi_{1/q}^2$.
Using that the sum of $p$ independent identically distributed 
$\chi_{1/q}^2$ random variables have a $\chi_{p/q}^2$ distribution
establishes the result.\qed
\end{proof}

\subsection{Generalized Marsaglia approach}\label{subsec:viaggd}
We restrict ourselves to the case when the number of degrees of freedom is 
rational, i.e.\/ $\nu=p/q$ with $p,q\in\mathbb N$. 
The algorithm for generating central $\chi_\nu^2$
samples is as follows.
\begin{algorithm}[Exact central chi-square samples]\label{alg:centralchisquare}
\begin{enumerate}
\item Generate $2q$ independent uniform random variables over $[-1,1]$:
$U=(U_1,\ldots,U_{2q})$.
\item If $\|U\|_{2q}<1$ continue, otherwise repeat Step~1. 
\item Compute $Z=U\cdot (-2\log\|U\|_{2q}^{2q})^{1/2q}/\|U\|_{2q}$. This gives
$2q$ independent $\text{N}(0,1,2q)$ distributed random variables $Z=(Z_1,\ldots,Z_{2q})$. 
\item Compute $Z_1^{2q}+\cdots+Z_{p}^{2q}\sim\chi_{p/q}^2$.
\end{enumerate}
\end{algorithm}
\begin{remark}
Note that if $p<2q$ then we can use the remaining $\text{N}(0,1,2q)$
random variables we generate in Step~3 the next time we need to generate
a $\chi_{p/q}^2(\lambda)$ sample. In practice we don't really need to
consider the case $p\geqslant2q$, but for the sake of completeness, we would
simply generate $p-2q$ more $\text{N}(0,1,2q)$ samples by repeating Steps~1--3.
\end{remark}
In Step~2, the probability of accepting $U_1,\ldots,U_{2q}$ is
given by the ratio of the volumes of $\mathbb S_{2q}(1)$ and $[-1,1]^{2q}$:
$P_{\mathrm{Mar}}\coloneqq\bigl(\Gamma(1/2q)/2q\bigr)^{2q}$.
Note for $q=1$, the probability of acceptance is 0.7854. Further as
$q\to\infty$ we have $P_{\mathrm{Mar}}\to\exp(-\gamma)\approx 0.5615$.
Here $\gamma$ is the Euler--Mascheroni constant and 
$\Gamma(z)\sim 1/z-\gamma$ as $z\to0^+$.

In practice we will need to generate a large number of samples. 
For the generalized Marsaglia polar method,
in each accepted attempt, we generate $2q$ generalized Gaussian 
random variables. Of these, $p$ random variables are used to generate 
a $\chi_{p/q}^2$ random variable. The number of attempts until the first 
success has a geometric distribution with mean $1/P_{\mathrm{Mar}}$. 
Hence the expected number of steps to generate $2q/p$ 
independent $\chi_{p/q}^2$ random variables is thus $1/P_{\mathrm{Mar}}$.

How does the acceptance rate of our central chi-squared
sampling method based on the generalized Marsaglia polar method,
for the case $\nu<2$, compare to the two leading acceptance-rejection methods?
These are the methods of Ahrens and Dieter~\cite{AD:chi} (also see
Glasserman~\cite[pp.~126--7]{Glass}) and the method of 
Marsaglia and Tsang~\cite{MarTsa}.

The acceptance-rejection algorithm of Ahrens--Dieter
is based on a mixture of the
prior densities $(\nu/2)\,x^{\nu/2-1}$ on $[0,1]$ and $\exp(1-x)$ on
$(1,\infty)$,
with weights $\mathrm{e}/(\mathrm{e}+\nu/2)$ and
$(\nu/2)/(\mathrm{e}+\nu/2)$,
respectively; here $\mathrm{e}=\exp(1)$. This method generates one
$\chi_\nu^2$
random variable with probability of acceptance
$P_{\mathrm{AD}}\coloneqq
(\nu/2)\,\Gamma(\nu/2)\,\mathrm{e}/(\nu/2+\mathrm{e})$.
In this method, the number of degrees of freedom $\nu$ can
be any real number. The expected number of attempts to generate $2q/p$
independent $\chi_{p/q}^2$ distributed
random variables is thus $(2q/p)\cdot (1/P_{\mathrm{AD}})$.
How do the expected number of attempts compare?
In other words, to generate $2q/p$ random variables, is
$1/P_{\mathrm{Mar}}\leqslant (2q/p)\cdot (1/P_{\mathrm{AD}})$?
Or equivalently, when does
$p/q\leqslant 2 P_{\mathrm{Mar}}/P_{\mathrm{AD}}$ hold? We
examine the right-hand side more carefully; set $z\coloneqq1/2q$, so
$0<z<1/2$. Then we have
\begin{equation*}
\frac{P_{\mathrm{Mar}}}{P_{\mathrm{AD}}}=
\bigl(z\,\Gamma(z)\bigr)^{1/z}\cdot\frac{\nu/2+\mathrm{e}}{(\nu/2)\,\Gamma(\nu/2)\,\mathrm{e}}.
\end{equation*}
Note $z$ and $\nu/2$ are independent. A lower bound for
$\bigl(z\,\Gamma(z)\bigr)^{1/z}$ is $\exp(-\gamma)\approx 0.5615$ for
$0<z<1/2$,
whilst a lower bound for
$(\nu/2+\mathrm{e})/((\nu/2)\,\Gamma(\nu/2)\,\mathrm{e})$
is $1$ for $0<\nu<2$. Hence $2 P_{\mathrm{Mar}}/P_{\mathrm{AD}}>1$ and
so for $p/q<1$, the expected number of attempts for the
generalized Marsaglia method is less than that for
the Ahrens--Dieter method.
We further note that the expected number of attempts for the generalized
Marsaglia method
to generate $2q/p$ chi-square samples is bounded by its value for $q=1$ and
the limit as $q\to\infty$, more precisely the expected number of attempts
is $1/P_{\mathrm{Mar}}\in (1.2732,\, 1.7811)$.
In contrast in the Ahrens--Dieter method, the expected number of
attempts to generate $2q/p$ samples is $(2q/p) (1/P_{\mathrm{AD}})$,
which is unbounded.

The expected number of steps is relevant in the context
of how often an algorithm is called. The generalized Marsaglia method
uses $2q$ uniform random variables each time it is called,
  while the algorithm of Ahrens--Dieter requires only 2
random variables. The expected number of random variables required
to generate $2q/p$ chi-square samples is thus $2q/P_{\mathrm{Mar}}$  and
$4q/p \cdot	1/P_{\mathrm{AD}}$, respectively. Since
\begin{equation*}
\frac{4q/p \cdot 1/P_{\mathrm{AD}}}{2q \cdot 1/P_{\mathrm{Mar}}}
= \frac{2}{p}\cdot\frac{P_{\mathrm{Mar}}}{P_{\mathrm{AD}}}
\end{equation*}
and since $P_{\mathrm{Mar}}<P_{\mathrm{AD}}$, we see that the expected required
number of input random variables is smaller for the
generalised Marsaglia method %than for the Ahrens--Dieter method 
for $p=1$ only, that is when the degrees of freedom can be written in the
form $p/q=1/2q$. %In Section  \ref{subsec:compacceff} we
%include a detailed comparison in terms of computational
%effort for these two methods.

Secondly, we compare the generalized Marsaglia method 
with the acceptance--rejection method of
Marsaglia and Tsang~\cite{MarTsa}. 
Their method is based on taking  a transformation
$h(X)= d(1+cX)^3$ on the set $\{X>-1/c\}$, where the distribution 
of $X$ is such that $h(X)$ has the required gamma distribution. 
Here $d=p/2q -1/3$ and $c=1/\sqrt{9d}$. 
The random variable $X$ can be sampled using an
acceptance--rejection method based on sampling a Normal random
variable. The acceptance probability is 
$P_{\mathrm{MT}}=\int_{-1/c}^\infty \mathrm{e}^{g(x)}\,\mathrm{d}x\,\cdot 1/\sqrt{2\pi}\cdot \bigl(1-\Phi(-1/c)\bigr)$, 
where
\begin{equation*}
g(x)= d \ln\bigl((1+cx)^3\bigr) -d (1+cx)^3+d
\end{equation*}
and where $\Phi$ denotes the standard Normal distribution
function. The algorithm of  Marsaglia--Tsang assumes that the gamma
parameter $p/2q\ge 1$. As noted there, a gamma random variable 
$\gamma(\alpha)$ with gamma parameter $\alpha <1$  can be generated by
$\gamma(\alpha)=\gamma(1+\alpha) U^{1/\alpha}$, where
$U\sim \text{U}(0,1) $. % is a uniform random variable.
The acceptance probability $P_{\mathrm{MT}}$ can be numerically
evaluated, for example, for $p/2q=1$ its value is 
$0.95167\cdot0.992847=0.944864$, while for $p/2q=2$ 
its value is $0.98166\cdot0.999946=0.98161$.
By analogy to the comparison with the algorithm of
Ahrens--Dieter above, we see that the expected number of steps to
generate $2q/p$ independent $\chi^2_{p/q}$ random is smaller for the 
generalized Marsaglia method compared with the algorithm by 
Marsaglia and Tsang for $p<q$---the regime of interest here. 
The expected number of uniform random variables required to do this
however is larger for the generalized Marsaglia method except if $p/q=1/2q$. 
We compare and discuss the numerical efficacy of both methods and of the 
algorithm of Ahrens and Dieter in Section~\ref{subsec:compacceff} below.

\subsection{Direct inversion}\label{subsec:viaggdi}
We restrict ourselves to case when the number of degrees of freedom 
$0<\nu<1$ is given to the first three decimal places and can thus
be expressed in the form
\begin{equation*}
\nu=\frac{p_5}{5}+\frac{p_{10}}{10}+\frac{p_{20}}{20}+
\frac{p_{50}}{50}+\frac{p_{100}}{100}+\frac{p_{200}}{200}+
\frac{p_{500}}{500}+\frac{p_{1000}}{1000}+\frac{p_{2000}}{2000},
\end{equation*}
for some $p_q\in\{0,1,2\}$ with $q\in\mathbb S$ where
$\mathbb S=\{5,10,20,50,100,200,500,1000,2000\}$.
For example, for $\nu=0.387$ we would have, since a $\text{N}(0,1,2q)$
random variable is needed to generate a $\chi_{1/q}^2$ one:
$p_{10}=1$, $p_{20}=1$, $p_{50}=2$, $p_{500}=1$, $p_{1000}=1$ and $p_{2000}=1$,
with the other $p_q$ coefficients equal to zero.
Of course in principle we can extend our method to degrees of freedom
$0<\nu<1$ given to any number of decimal places (see the discussion
at the end of this section). In Appendix~\ref{app:coeffs} we quote
the coefficients and parameter values for Pad\'e and Chebychev
approximants required for direct inversion for the
values of $q$ in $\mathbb S$. The algorithm for generating 
central $\chi_\nu^2$ samples, for any $\nu$ given to the first 
three decimal places, using direct inversion of 
generalized Gaussian samples is as follows. 
\begin{algorithm}[High accuracy central chi-square samples]\label{alg:CChisqdi}
\begin{enumerate}
\item For each $q\in\mathbb S$, generate $p_q$
independent uniform random variables over $[0,1]$: $(U_{1,q},\ldots,U_{p_q,q})$,
i.e.\/ a total of $\sum_{q\in\mathbb S}p_q$ independent uniform random variables.
\item For each $q\in\mathbb S$, use $(U_{1,q},\ldots,U_{p_q,q})$ and the direct inversion
algorithm in Appendix~\ref{app:di} to generate $p_q$ generalized
Gaussian samples $(Z_{1,q},\ldots,Z_{p_q,q})$.
\item Compute $\sum_{q\in\mathbb S}\sum_{k=1}^{p_q} Z_{k,q}^{2q}\sim\chi_{\nu}^2$.
\end{enumerate}
\end{algorithm} 
\begin{remark}
We chose the set $\mathbb S$ and decomposition of $\nu$ above
for efficiency and convenience. We could achieve greater
efficiency by decomposing $\nu$ more finely, but this would 
increase the number of Pad\'e and 
Chebychev approximants we need to compute and store. %This in turn
%would start to impinge on the simplicity and preparation
%requirements of the method. 
\end{remark}

\subsection{Comparison}\label{subsec:compacceff}
The natural question is how do the algorithms perform 
in practice? Two issues immediately surface. The first
is that the generalized Marsaglia approach is restricted 
to rational numbers. In practical applications this is 
not restrictive as all finite precision arithmetic is
in principle rational arithmetic. The second is that
the direct inversion approach, as we have implemented it here,
only allows for three decimal places. However,  this is not a restriction 
either, as we can in principle construct additional approximations to the inverse 
generalized Gaussian distribution function for larger
values of $q$. Indeed for  $q=10^4$, we computed
a $(3,1)$ Pad\'e approximation for the central region,
a $(4,5)$ Pad\'e approximation for the middle region
and a degree $10$ Chebychev approximation for the 
tail that guarantee accuracy of $10^{-9}$ across
all regions for the inverse generalized Gaussian distribution function.
We also note that parameter values 
are typically
determined by calibration and quoted to only 
$2$ or $3$ significant figures. 

Two leading gamma random variable acceptance-rejection 
sampling methods are those of Ahrens and Dieter~\cite{AD:chi} and 
Marsaglia and Tsang~\cite{MarTsa}. In Figure~\ref{fig:cputime} 
we compare CPU times needed to generate $10^5$ samples versus 
the number of degrees of freedom $\nu$ using the
generalized Marsaglia, Ahrens--Dieter, 
Marsaglia--Tsang and direct inversion methods. 
The values for the degrees of freedom chosen are
$\nu=(1+m)\cdot 10^{-4}$ for $m=0,1,\ldots,1000$ and
$\nu=0.101+m\cdot 10^{-3}$ for $m=0,1,\ldots,299$---we
omitted CPU times for the direct inversion method involving 
the fourth decimal place. The CPU times were generated
using \emph{compiled} Matlab code to better reflect
their potential practical implementations.
We observe the Ahrens--Dieter and Marsaglia--Tsang 
acceptance-rejection methods roughly require the same 
CPU time to generate central $\chi_\nu^2$ samples for 
any of these values of $\nu$. The Ahrens--Dieter method
also appears to be slightly more efficient.
The generalized Marsaglia approach shows more variation 
in the CPU time required.
In particular for example, for values of $\nu$ equal to 
$3$, $6$, $7$ and $9$ times $10^{-m}$ for all $m=2,3,4$, it takes
longer to generate central $\chi_\nu^2$ samples than for the other $\nu$ values. 
This is due to the fact that as rational numbers, with denominators
as powers of $10$, they do not simplify nicely to what might be 
considered the optimal format for sampling with this method,
namely $1/2q$ (see also Section \ref{subsec:viaggd}). For values of $\nu$ which cannot be reduced to this
optimal format, we need to sum over a number of 
generalized Gaussian samples to produce a central $\chi_\nu^2$
sample. However any decimal with a finite number of significant
figures can be written as the sum of fractions of powers of $10$.
Further a central $\chi_\nu^2$ random variable can be constructed by adding 
independent central $\chi_{\nu_i}^2$ random variables for which $\nu_1+\cdots+\nu_k=\nu$.
Indeed for all the other values of $\nu$ shown in Figure~\ref{fig:cputime}
we generated the central $\chi_\nu^2$ samples by the generalized Marsaglia 
approach by adding $\chi_{\nu_i}^2$ samples where the $\nu_i$ are fractions 
of powers of $10$ that generate each significant figure.
We observe that with this decomposition technique the 
generalized Marsaglia approach is overall marginally slower
than the Ahrens--Dieter and Marsaglia--Tsang methods. 
However this is partly an artifact of the requirement 
to add multiple $\chi_{\nu_i}^2$ samples to generate 
$\chi_{\nu}^2$. This could be alleviated by a finer
decomposition or more directly by implementing the 
generalized Marsaglia approach for the given 
rational $\nu$ (as we will see in Section~\ref{sec:simulations}). 
For the direct inversion $\chi_\nu^2$ 
sampling we used the even finer decomposition into 
fractions of $q\in\mathbb S$. We observe that overall, 
its performance is superior to the other methods. 

\begin{figure} 
  \begin{center}
  \includegraphics[width=10.0cm,height=8.0cm]{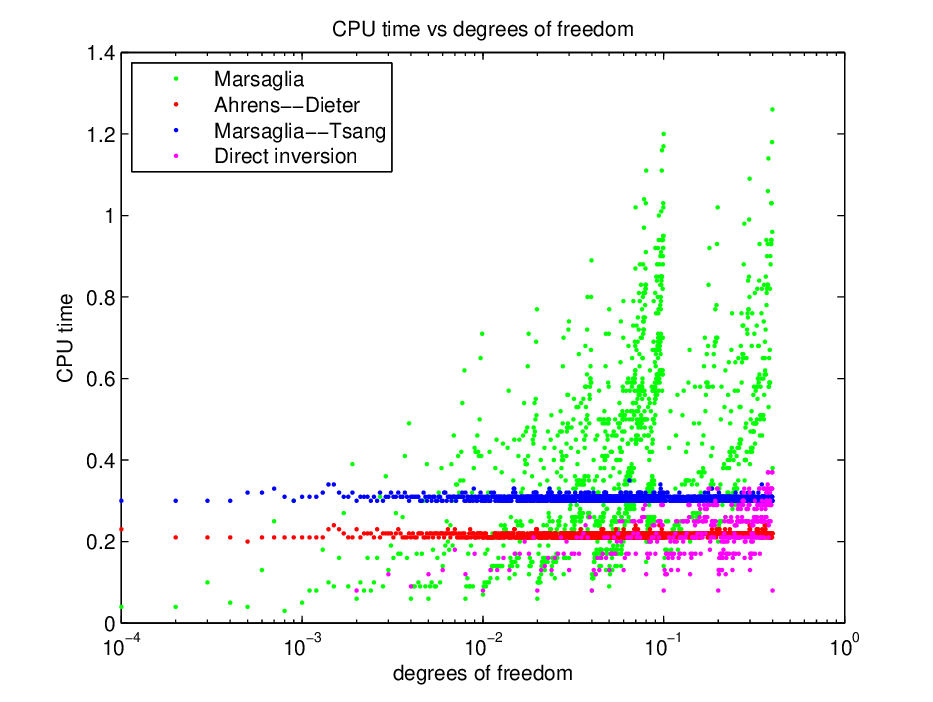} 
  \end{center}
\caption{CPU time versus the number of degrees of freedom $\nu$ using the
generalized Marsaglia, Ahrens--Dieter and Marsaglia--Tsang acceptance-rejection
methods, as well as the direct inversion method. 
The ordinate shows the CPU time needed to generate $10^5$ samples,
simultaneously for two sets of $\nu$ abscissa values given to three significant figures,
namely $\nu=(1+m)\cdot 10^{-4}$ for $m=0,1,\ldots,1000$ and
$\nu=0.101+m\cdot 10^{-3}$ for $m=0,1,\ldots,299$. We have omitted CPU times
for the direct inversion method involving the fourth decimal place.}
\label{fig:cputime}
\end{figure} 

\begin{remark} 
We could equivalently construct analogous Pad\'e and polynomial
approximations for the inverse $\chi^2_\nu$ distribution function.
Indeed we could generate tables analogous to those in 
Appendix~\ref{app:coeffs} for values of 
$\nu=0.4, 0.2, 0.1, 0.04, 0.02, 0.01,\ldots$ and so forth.
We chose to base our method on the generalized Gaussian distribution
as: (i) we could use it to sample from both the generalized Gaussian 
and central chi-square distributions simultaneously and we thus
afforded greater flexibility; (ii) we used the robust and highly
effective Beasley--Springer--Moro method for the inverse Gaussian
distribution as a starting point, and (iii) the behaviour of the
generalized Gaussian distribution in the limit of large $q$---to
a uniform $\text{U}(-1,1)$ distribution---was qualitatively and 
quantitatively appealing. However,
suppose one knew beforehand that for a known fixed number of 
degrees of freedom $\nu$ many $\chi^2_\nu$ samples would be required.
In such a scenario it may be worth expending preparation 
effort to generate Pad\'e and polynomial approximations
for that value of $\nu$, as one could then relatively efficiently
sample using the inverse $\chi^2_\nu$ distribution function approximation. 
\end{remark}
\begin{remark} 
Since the direct inversion method computes the 
inverse generalized Gaussian distribution function
to very high accuracy (and in principle to machine error), 
it can be combined with variance-reducing Monte Carlo techniques, 
e.g.~antithetic variates or conditional Monte Carlo where appropriate 
for the problem under consideration.
A further advantage of direct inversion methods 
for the use of  \emph{quasi-Monte Carlo} simulation is that
exactly one random input variable is required to
generate one sample of the target distribution.
See Chapter~2 in Glasserman~\cite{Glass} for a more detailed discussion. 
%can be used with  for all problems for which antithetic variates are
%variance--reducing, e.g.~when the output is monotone as a function of
%the input generalized Gaussian random variable. 
\end{remark}

%\begin{remark}$\ast$
%We need to discuss the halving sampling issue of our direct inversion method. 
%\end{remark}

\subsection{Non-central chi-square sampling}\label{subsec:NCChisq}
We can generate non-central chi-square samples from
central chi-square samples as follows. Following 
Johnson~\cite{Jo1959} and Siegel~\cite{Siegel}
any $\chi_\nu^2(\lambda)$ random variable can be
decomposed as $\chi_\nu^2(\lambda)\sim\chi_0^2(\lambda)+\chi_\nu^2$.
Here $\chi_0^2(\lambda)$ random variable can be generated
by choosing a random variable $N$ from a 
Poisson distribution with mean $\lambda/2$, and
then generating a central $\chi_{2N}^2$ sample. Also
see Broadie and Kaya~\cite{BK} for more details.
Hence we are left with the problem of how 
to sample from a $\chi_{2N}^2$ distribution---we use
either of the two methods of the last section to
sample from the $\chi_\nu^2$ distribution. The following
algorithm produces a $\chi_{\nu}^2(\lambda)$ sample.
\begin{algorithm}[Non-central chi-square samples]\label{alg:noncentral}
\begin{enumerate}
\item Use Algorithm~\ref{alg:centralchisquare} or~\ref{alg:CChisqdi}
to generate a $\chi_\nu^2$ random variable $Z$.
\item Generate a Poisson distributed random variable $N$ with mean $\lambda/2$.
\item Generate $N$ independent uniform $\text{U}(0,1)$ random variables $U_1,\ldots,U_{N}$.
\item Compute $-2\bigl(\log(U_1)+\cdots+\log(U_N)\bigr)+Z\sim\chi_{\nu}^2(\lambda)$.
\end{enumerate}
\end{algorithm}
For small non-centrality $\lambda\leqslant10$, 
we generate the Poisson random variable using the
exact direct inversion method (inverse transform method) in 
Glasserman~\cite[p.~128]{Glass}. (An alternative method
for small non-centrality is the acceptance-rejection method
found in Knuth~\cite[p.~137]{Knuth}.) 
When the non-centrality parameter $\lambda>10$ we could
use the Normal approximation from Fishman~\cite[p.~212, paragraph~3]{Fish}: 
$\max\{0,\text{floor}(\lambda+1/2+\sqrt\lambda\,Y)\}$,
where $Y$ is a standard Normal random variable. However
we are endeavouring to retain accuracy as far as possible and
prefer to avoid such approximations. There are several other 
notable methods for generating Poisson random variables in this 
parameter regime, in particular that of Ahrens and Dieter~\cite{AD} 
and the PRUA$^*$ method found in Fishman~\cite[p.~214]{Fish};
but these are acceptance-rejection methods.
%
%On average this algorithm requires the generation
%of $\mu+1$ uniform variates. When the Poisson mean $\mu=\lambda/2$ is large, 
%this algorithm is not efficient. 

Large non-centrality $\lambda>10$ is relevant to our
applications to the Heston model: the Poisson mean is inversely proportional to
the discretization stepsize, which  is required to be small when for example one 
wishes to price path-dependent derivatives. We thus propose the 
following approach to chi-square sampling when the 
non-centrality parameter $\lambda$ is large, 
say larger than a critical value $\bar{\lambda}$,  
modifying Algorithm~\ref{alg:noncentral} as follows.
The Poisson variable $N$ can be written as a sum of two
independent Poisson random variables $\bar{N}$  and $P$ with mean
$\bar{\lambda}/2$ and  mean
$\lambda/2-\bar{\lambda}/2$, respectively. The chi-square
distribution can be represented as
\begin{equation*}
\chi^2_\nu(\lambda)\sim\chi_{\nu+2N}^2\sim\chi_{\nu+2\bar{N}+2P}^2\sim\chi^2_{\nu+2\bar{N}}(\lambda-\bar{\lambda}).
\end{equation*}
We sample $\bar{N}$ from the Poisson distribution with parameter
$\bar{\lambda}/2$ using the direct inversion method
in Glasserman~\cite[p.~128]{Glass} (mentioned above for small non-centrality). 
If $\bar{N}\not=0$, then the $\chi^2_{\nu+2\bar{N}}(\lambda-\bar{\lambda})$ variable
can be represented as a sum of a $\chi^2_\nu$ variable and an independent
$\chi^2_{2\bar{N}}(\lambda-\bar{\lambda})$ variable
$\chi^2_{\nu+2\bar{N}}(\lambda-\bar{\lambda})\sim
\chi^2_{\nu} + \chi^2_{2\bar{N}}(\lambda-\bar{\lambda})$. 
A sample from this distribution can be generated efficiently by
sampling $\bar{N}-1$ independent  uniform $\text{U}(0,1)$ random
variables, say $U_1,\ldots,U_{\bar{N}-1}$, two independent  standard normal random
variables, say $V_1$ and $V_2$,  and an independent $\chi_\nu^2$ random variable, say $Z$, 
using Algorithm~\ref{alg:centralchisquare}. Then 
$-2(\log U_1+\cdots+\log U_{\bar{N}-1})+V_1^2+\bigl(V_2+\sqrt{\lambda-\bar{\lambda}}\bigr)^2
+Z\sim\chi_{\nu+2\bar{N}}(\lambda-\bar{\lambda})$.
If $\bar{N}=0$, then we have to sample a 
$\chi^2_\nu(\lambda-\bar{\lambda})$ random variable, but now with a
non-centrality parameter   $\lambda-\bar{\lambda}<\lambda$. 
If $\lambda-\bar{\lambda}\leqslant\bar{\lambda}$,
then the direct inversion method in Glasserman~\cite[p.~128]{Glass} 
is an efficient method to sample from this distribution. 
If  $\lambda-\bar{\lambda}$ is larger than $\bar{\lambda}$,  then
we  repeat this process until the sample of the Poisson random variable with
mean $\bar{\lambda}$ returns a non-zero value or until the non-centrality parameter
is smaller than $\bar{\lambda}$, whichever comes first. 
To summarize, the algorithm is as follows.
\begin{algorithm}[Chi-square samples for large non-centrality parameter]\label{alg:Poisson2}
\begin{enumerate}
\item If $\lambda >\bar{\lambda}$, generate a Poisson random variable $\bar{N}$
with mean $\bar{\lambda}/2$ using direct inversion or the method in Knuth.
\item 
\begin{enumerate}
\item If $\bar{N}\not=0$
\begin{enumerate}
\item  generate $\bar{N}-1$  independent  uniform $\text{U}(0,1)$ random
variables, say $U_1,\ldots$, $U_{\bar{N}-1}$, two independent  standard normal random
variables, say $V_1$ and $V_2$, and use Algorithm \ref{alg:centralchisquare} 
to generate an independent $\chi^2_{\nu}$ random variable $Z$.
\item Compute $-2(\log U_1+\cdots+\log U_{\bar{N}-1})+V_1^2
+\bigl(V_2+\sqrt{\lambda-\bar{\lambda}}\bigr)^2+Z\sim\chi_{\nu}^2(\lambda)$.
\end{enumerate}
\item If $\bar{N}=0$, set $\lambda\leftarrow\lambda-\bar{\lambda}$. 
\begin{enumerate}
\item If $\lambda > \bar{\lambda}$, repeat from Step 1.
\item If $\lambda\leqslant\bar{\lambda}$, use 
Algorithm~\ref{alg:noncentral} to generate  
an independent $\chi^2_{\nu}(\lambda)$ random variable.
\end{enumerate}
\end{enumerate}
\end{enumerate}
\end{algorithm}
Algorithm~\ref{alg:Poisson2} provides an exact method to sample a
$\chi_\nu^2(\lambda)$ random variable for a large non-centrality parameter. 
Note that $P\{\bar{N}\not=0\}=1-\exp(-\bar{\lambda}/2)$. The expected
number of iterations to generate a chi-square sample is thus the minimum of
$1/\bigl(1-\exp(-\bar{\lambda}/2)\bigr)$ and $\lfloor\lambda/\bar{\lambda} \rfloor$.
However more importantly, note that the number of steps is bounded
above by $\lfloor\lambda/\bar{\lambda}\rfloor$. In our applications,
we set $\bar{\lambda}=20$. Thus the probability of performing
Step~2(b) in Algorithm~\ref{alg:Poisson2} is $P\{\bar{N}=0\}\approx10^{-5}$,
the probability  of repeating Step~2(b) is of order $10^{-9}$. The probability
of repeating Step~2(b) once more is of order $10^{-13}$. In most applications
this can be considered negligibly small. Thus practically, Algorithm~\ref{alg:Poisson2} 
can be stopped after a maximum of 2 iterations (beyond that a positive value 
could by arbitrarily assigned if required).
%any non-zero integer value can be assigned to $\bar{N}$, and in our context of sampling from
%the chi-square distribution, any non-negative value could be assigned to
%the chi-square random variable). 
Alternatively, we could  enforce an upper bound on the total number of random variables 
required here, by using the Andersen--Patnaik matched normal squared approximation for
non-centrality values $\lambda$ that are very large, i.e.\/ large enough so
as not to compromise the accuracy of the direct inversion method we have
carried thus far.
\begin{remark}
When the number of degrees of freedom $\nu$ is one or greater,
the decomposition $\chi_\nu^2(\lambda)\sim\chi_1^2(\lambda)+\chi_{\nu-1}^2$
radically simplifies the non-central chi-square simulation process.
A $\chi_1^2(\lambda)$ random variable is straightforwardly generated by squaring an
appropriately mean-shifted standard Normal random variable. If $\nu-1<1$
then a central $\chi_{\nu-1}^2$ random variable can be generated using
the methods described above. Whilst if $\nu-1\geqslant1$, then 
we can decompose 
$\chi_{\nu-1}^2\sim\chi_{\lfloor\nu-1\rfloor}^2+\chi_{\nu-1-\lfloor\nu-1\rfloor}^2$.
The component involving the integer part of $\nu-1$, i.e.\/ $\lfloor\nu-1\rfloor$, 
can be simulated by taking the sum of the logarithms of 
$\lfloor\nu-1\rfloor$ uniform random variables 
(by analogy with Step~4 in the Algorithm~\ref{alg:noncentral} above).
The component involving the remaining fractional degrees of
freedom $\nu-1-\lfloor\nu-1\rfloor$ can again be simulated using the 
methods described above. 
\end{remark}

\section{Application: the Cox-Ingersoll-Ross process and Heston model}
\label{sec:simulations}
We illustrate the accuracy of our new methods for chi-square sampling of the CIR process
applied to the Heston model. 

\subsection{The CIR process}\label{subsec:CIRsimulations}
The mean-reverting square-root process or \emph{Cox-Ingersoll--Ross (CIR) process} was 
first used in a financial context by Cox, Ingersoll and Ross~\cite{CIR} to model the 
short rate of interest and has been applied in numerous financial applications 
since. It can be expressed in the form
\begin{align*}
\mathrm{d}V_t=&\;\kappa(\theta-V_t)\,\mathrm{d}t
+\varepsilon\sqrt{V_t}\,\mathrm{d}W^1_t,
\end{align*} 
where $W^1$ is a Wiener process and  $\kappa$, $\theta$
and $\varepsilon$ are positive constants. 
It is a mean-reverting process with mean $\theta$, 
rate of convergence $\kappa$ and square root diffusion
scaled by $\varepsilon$.
By the Yamada condition this model has a unique strong solution.
Interestingly, though the explicit form of the solution as a function of the driving
Wiener process $W^1$ is not known, its transition probability is explicitly given 
as a scaled %linear transformation of a 
non-central chi-square distribution. 
%In particular, the CIR process $V$ is non-negative.
We define the \emph{degrees of freedom} for this process to be
$\nu\coloneqq4\kappa\theta/\varepsilon^2$.
When $\nu\in\mathbb N$ the process $V_t$ can be reconstructed
from the sum of squares of $\nu$ Ornstein--Uhlenbeck processes; 
hence the label of degrees of freedom. 
When $\nu<2$ the zero boundary is attracting and attainable, while when
$\nu\geqslant2$, the zero boundary is non-attracting. In particular, 
the CIR process is non-negative. These 
properties are immediate from the Feller boundary criteria, 
see Feller~\cite{Feller}. These are based on inverting 
the associated stationary elliptic Fokker--Planck operator, 
with boundary conditions, and can be found for example in 
Karlin and Taylor~\cite{KT}. 

Here we focus on the challenge of $\nu<2$ and in particular
cases when $\nu\ll\,1$. %typical of FX markets (Andersen~\cite{An}).
Importantly, though the zero boundary is attracting and attainable,
it is strongly reflecting---if the process reaches zero it
leaves it immediately and bounces back into the positive 
domain---see Revuz and Yor~\cite[p.~412]{RY}. We detailed
in the introduction how this case is a major obstacle, 
particularly for direct discretization methods. A comprehensive
account of direct discretization methods can
be found in Lord, Koekkoek and Van Dijk~\cite{LKVD}. %to where the reader is referred. 
The \emph{full truncation method} proposed by Lord, Koekkoek 
and Van Dijk allows the variance process to be negative over successive
timesteps---when the variance evolves deterministically with
an upward drift of $\kappa\theta$ and the volatility of the 
price process is taken to be zero. Andersen~\cite{An} and
Haastrecht and Pelsser~\cite{HP} complete thorough comparisons
with full truncation method of Lord, Koekkoek and Van Dijk.

The method we propose follows the lead of Broadie and Kaya~\cite{BK}
and is based on simulating the known transition probability density 
for the Cox--Ingersoll--Ross process. 
We quote the following form for this transition density,
that can be found in Cox, Ingersoll and Ross~\cite{CIR}, from a
proposition in Andersen~\cite{An}.
\begin{proposition}
Let $F_{\chi_\nu^2(\lambda)}(z)$ be the cumulative distribution
function for the non-central chi-squared distribution with
$\nu$ degrees of freedom and non-centrality parameter $\lambda$:
\begin{equation*}
F_{\chi_\nu^2(\lambda)}(z)=
\frac{\exp(-\lambda/2)}{2^{\nu/2}}\sum_{j=0}^\infty
\frac{(\lambda/2)^j}{j!2^j\Gamma(\nu/2+j)}
\int_0^z\xi^{\nu/2+j-1}\exp(-\xi/2)\,\mathrm{d}\xi.
\end{equation*}
Set $\nu\coloneqq4\kappa\theta/\varepsilon^2$ and define 
$\eta(h)\coloneqq 4\kappa\,\exp(-\kappa h)/
\varepsilon^2\bigl(1-\exp(-\kappa h)\bigr)$,
where $h=t_{n+1}-t_n$ for distinct times $t_{n+1}>t_n$.
Set $\lambda\coloneqq V_{t_n}\cdot\eta(h)$.
Then conditional on $V_{t_n}$, $V_{t_{n+1}}$ is distributed
as $\exp(-\kappa h)/\eta(h)$ times a non-central 
chi-squared distribution with $\nu$ degrees of freedom and
and non-centrality parameter $\lambda$,
i.e.\ 
\begin{equation*}
\mathbb P\bigl(V_{t_{n+1}}<x~\big|~V_{t_n}\bigr)=
F_{\chi_\nu^2(\lambda)}\bigl(x\cdot\eta(h)/\exp(-\kappa h)\bigr).
\end{equation*}
\end{proposition}
Hence for Cox--Ingersoll--Ross sampling from timestep
$t_n$ to $t_{n+1}$, we set $\lambda=V_{t_n}\cdot\eta(h)$ 
and compute/approximate 
$V_{t_{n+1}}=\chi_\nu^2(\lambda)\cdot\exp(-\kappa h)/\eta(h)$.

\begin{figure} 
  \begin{center}
  \includegraphics[width=8.0cm,height=5cm]{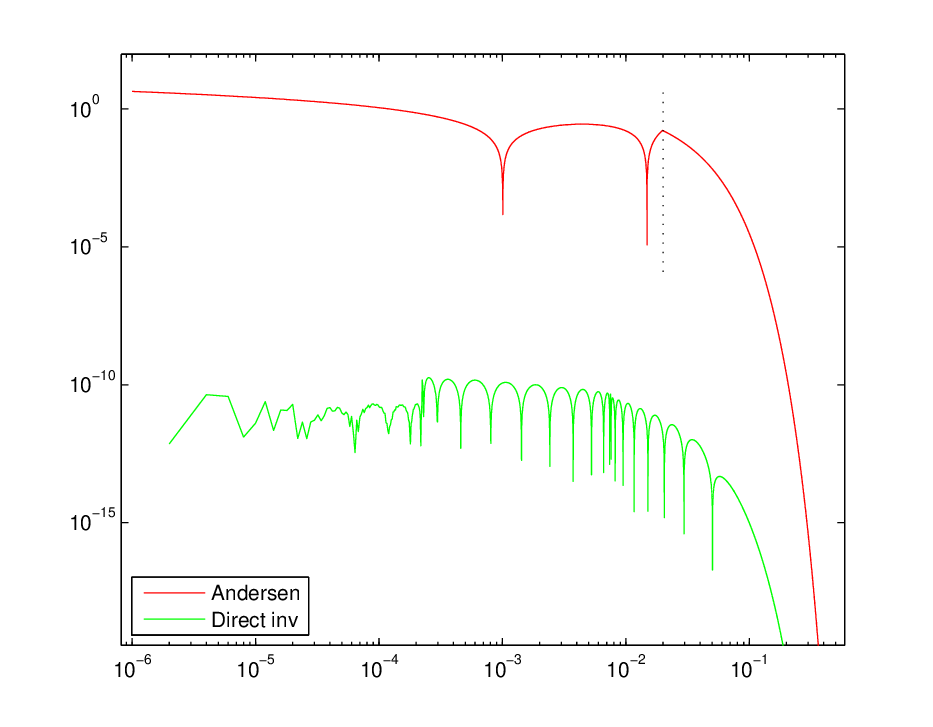}
  \end{center}
\caption{The absolute error in the direct inversion and Andersen approximations
of the inverse chi-square $\chi_\nu^2$ distribution function for $\nu=0.02$.
To the right of the vertical dotted line the Andersen approximation is identically
zero and the curve shown represents the absolute difference between zero and 
the exact inverse distribution function.}
\label{fig:icdferr}
\end{figure} 

\begin{figure} 
  \begin{center}
  \includegraphics[width=6.0cm,height=5.0cm]{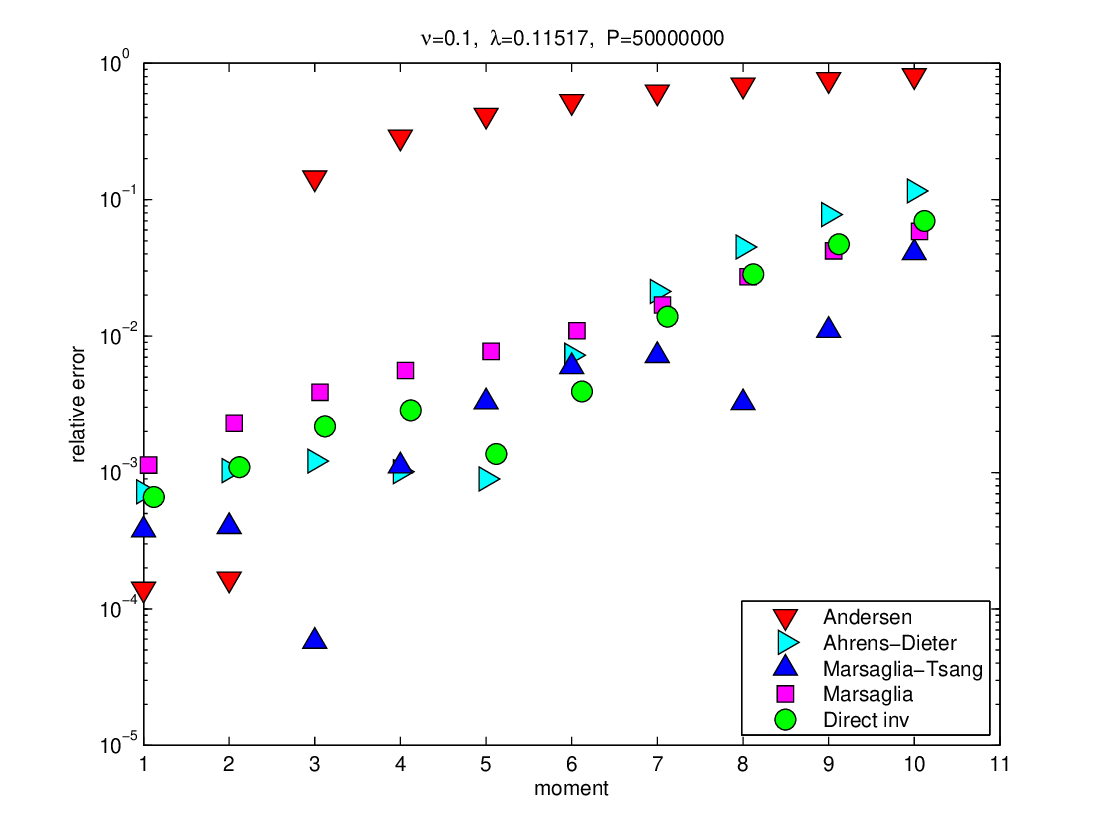}
  \includegraphics[width=6.0cm,height=5.0cm]{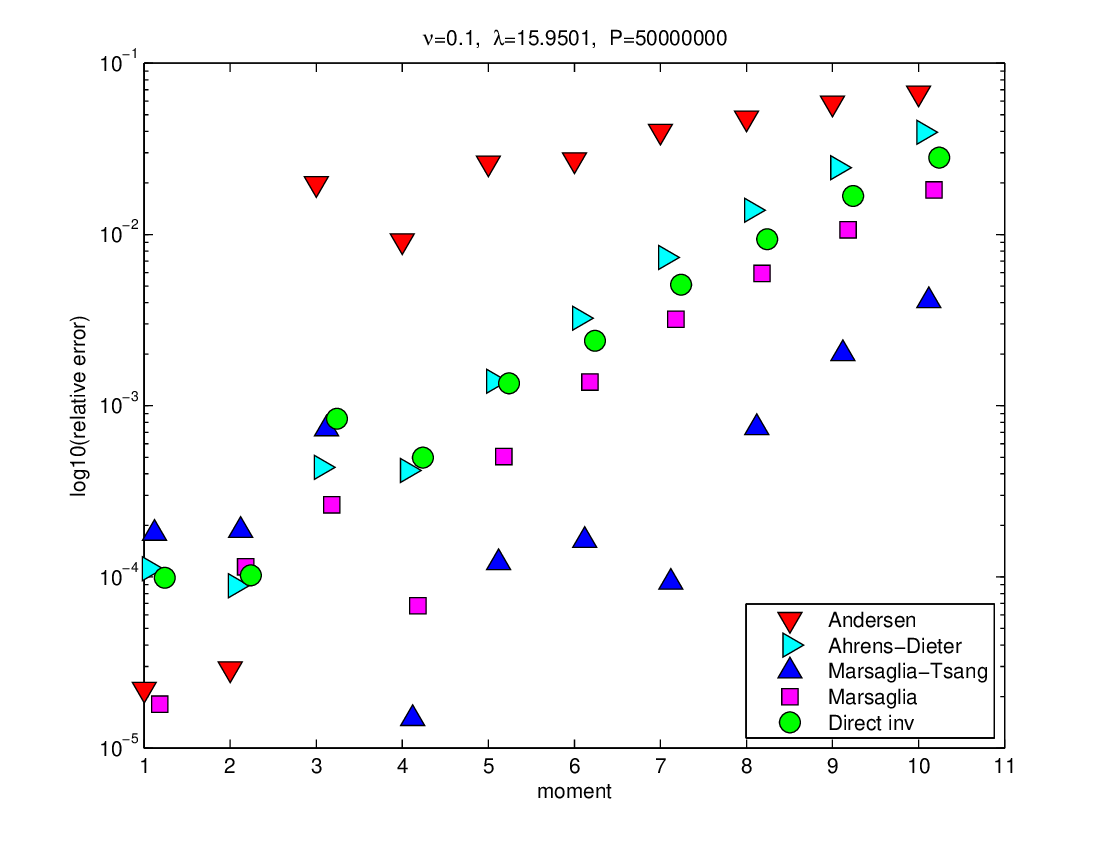}\\
  \includegraphics[width=6.0cm,height=5.0cm]{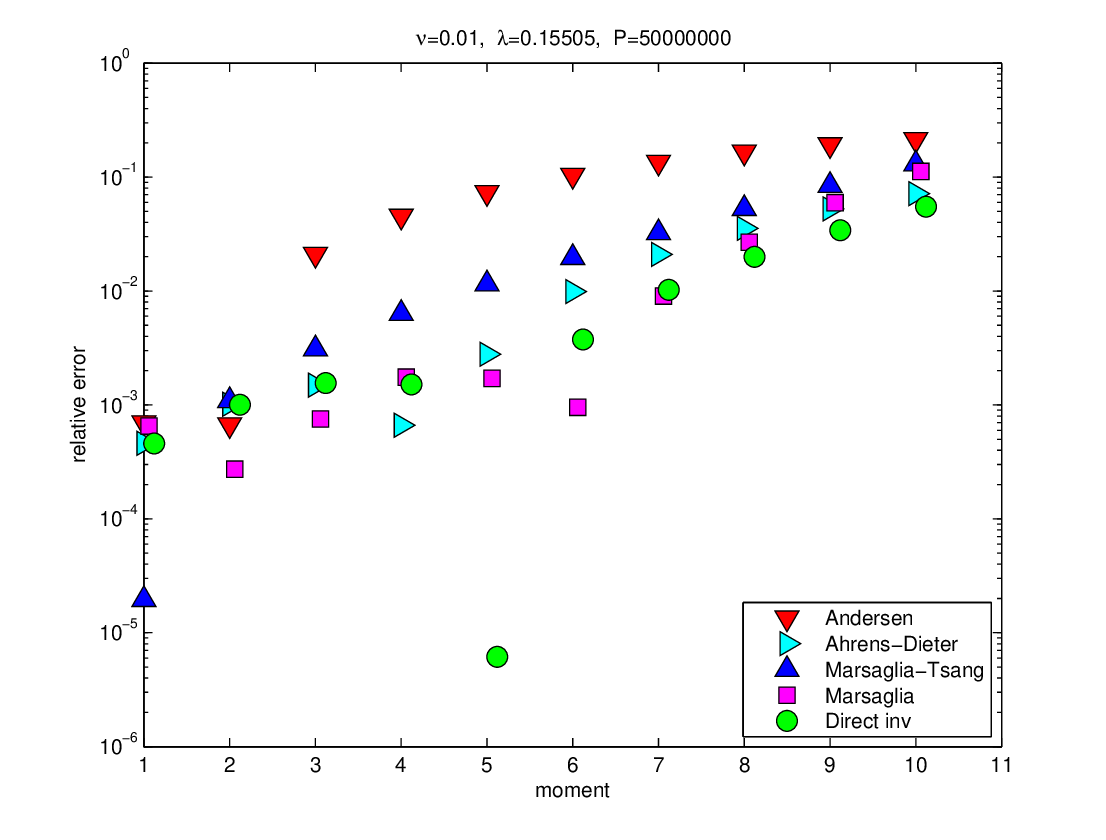}
  \includegraphics[width=6.0cm,height=5.0cm]{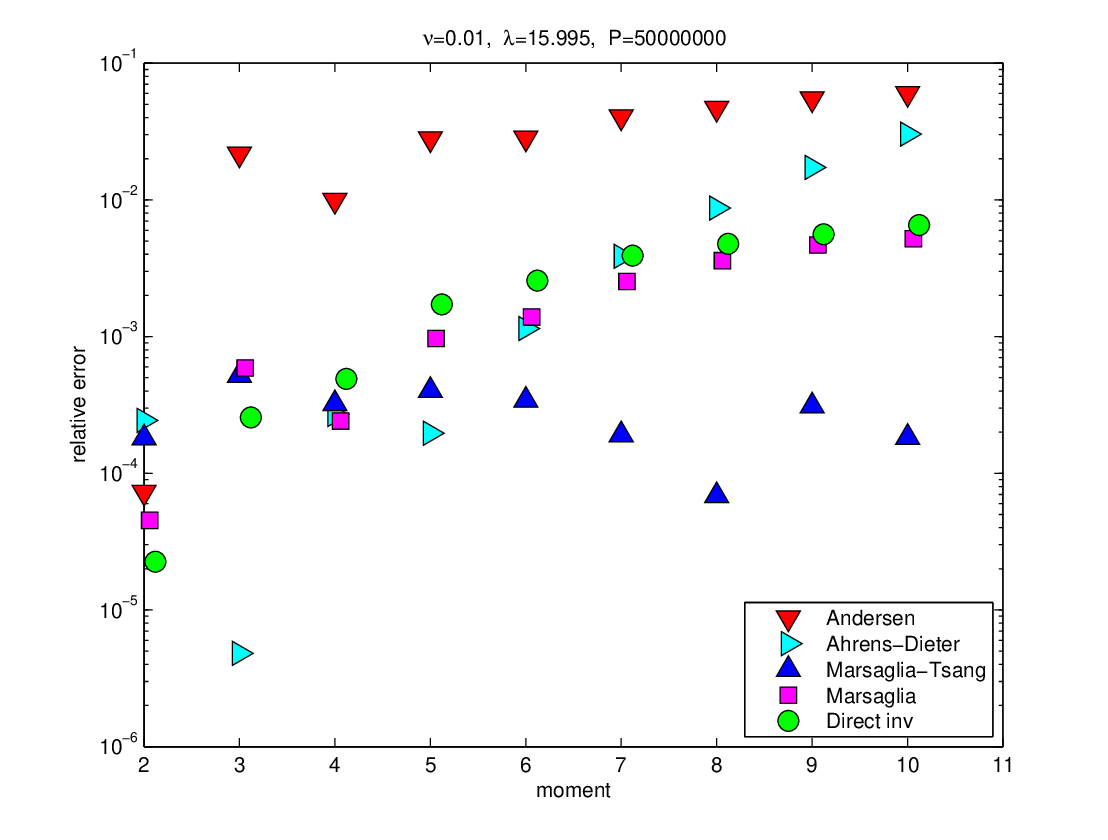}\\
  \includegraphics[width=6.0cm,height=5.0cm]{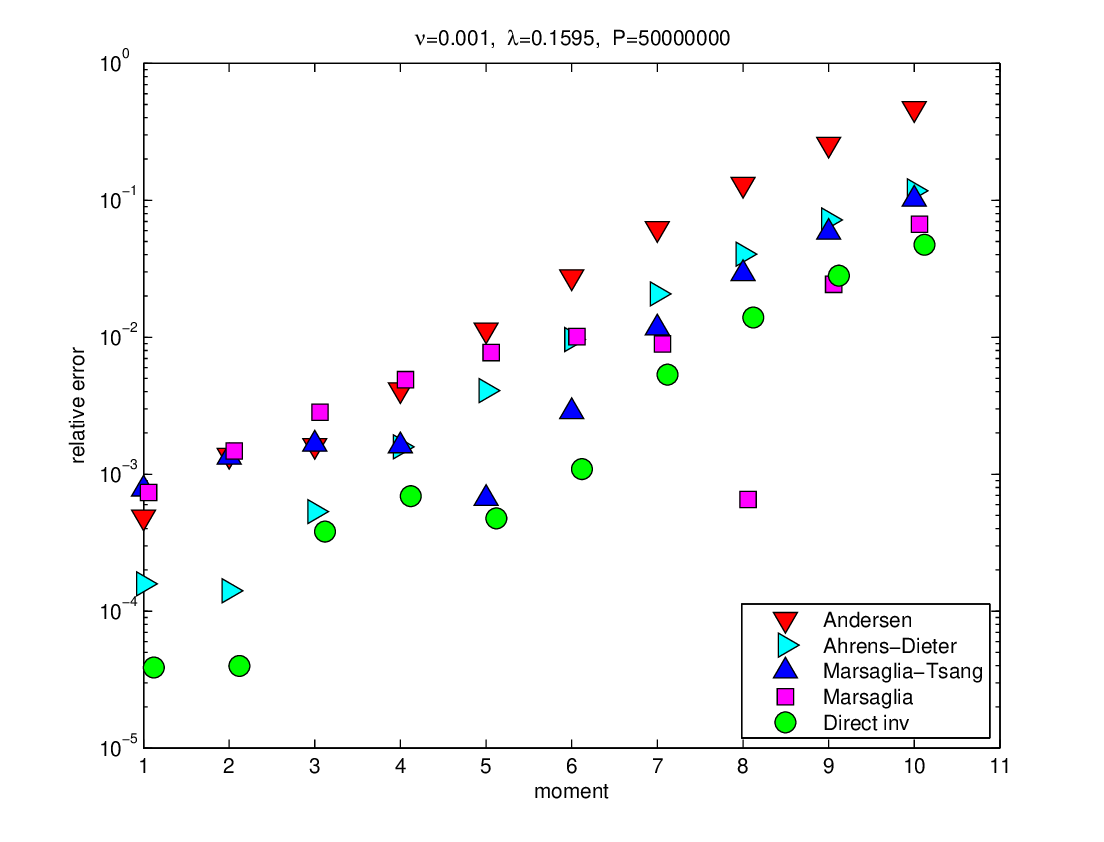}
  \includegraphics[width=6.0cm,height=5.0cm]{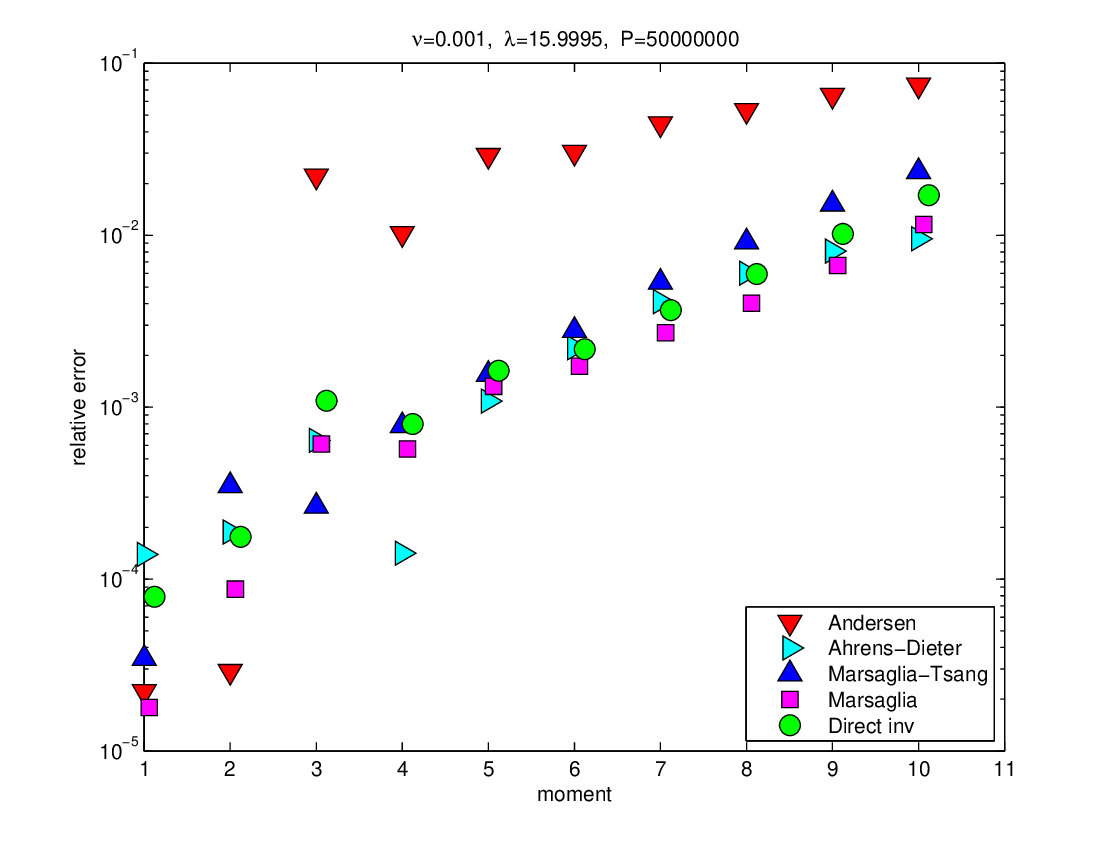}\\
  \end{center}
\caption{The relative error in the first through tenth moments of the 
non-central chi-square sampling methods shown. The six
panels correspond to the three values of the degrees
of freedom $\nu=0.1,0.01,0.001$ (top to bottom), and the
values for the non-centrality 
$\lambda=0.11517,15.9501,0.15505, 15.995,0.1595,15.9995$
(left to right, then top to bottom). The number of samples used
in each case is $5\times 10^7$.}
\label{fig:momentserr}
\end{figure} 

\begin{figure} 
  \begin{center}
  \includegraphics[width=6.0cm,height=5.0cm]{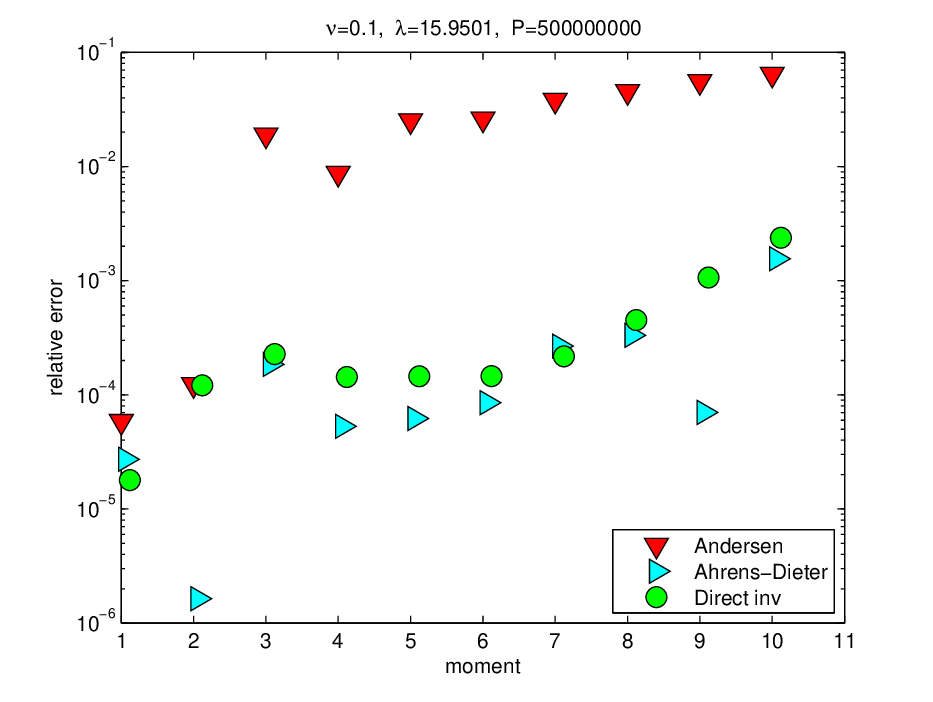}
  \includegraphics[width=6.0cm,height=5.0cm]{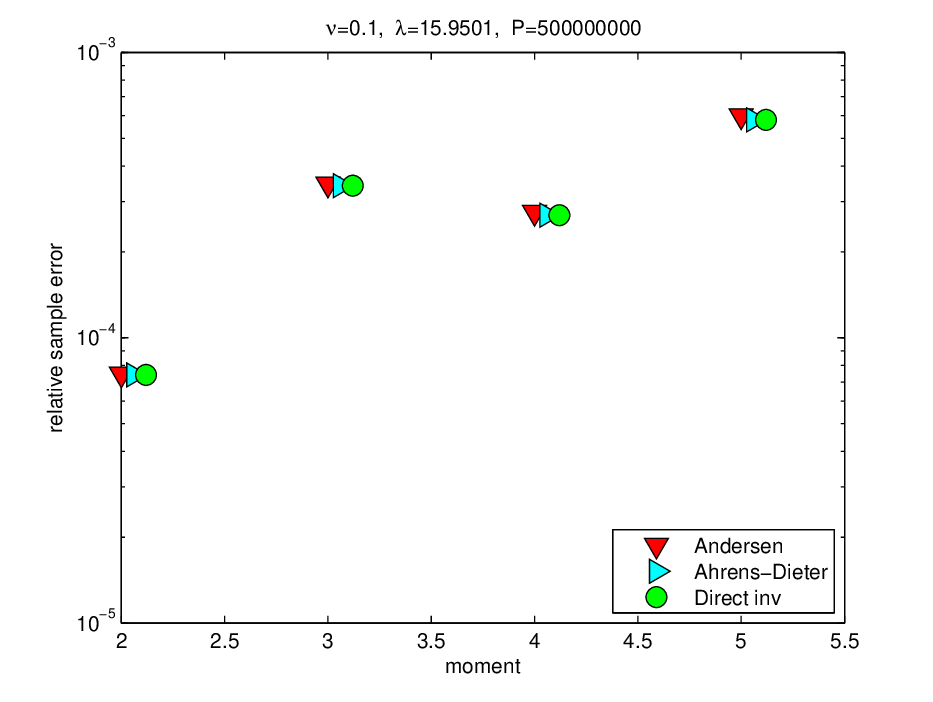}
  \includegraphics[width=6.0cm,height=5.0cm]{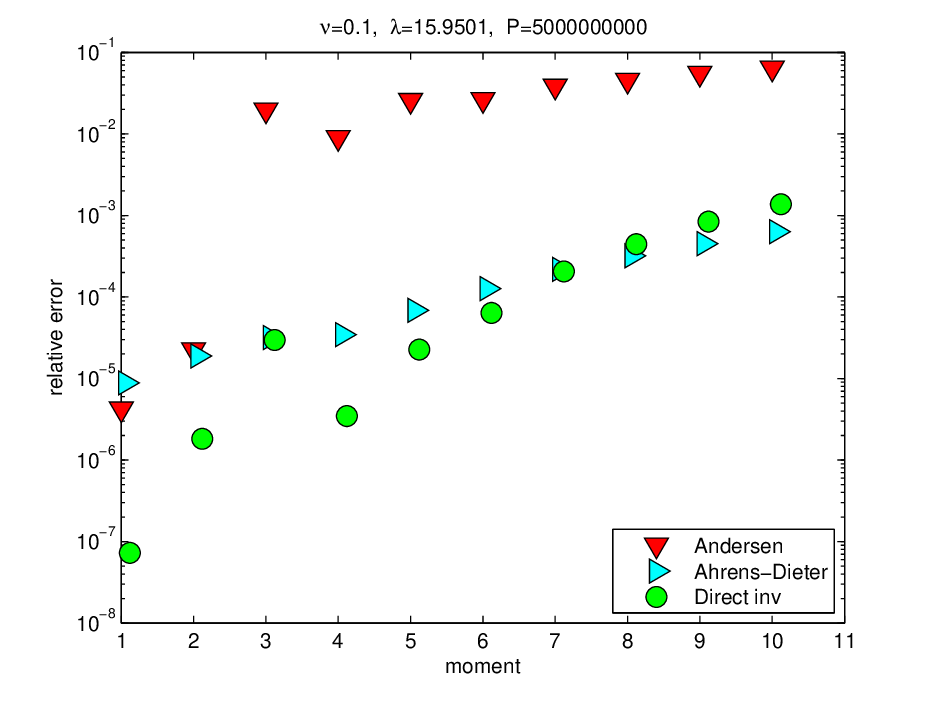}
  \includegraphics[width=6.0cm,height=5.0cm]{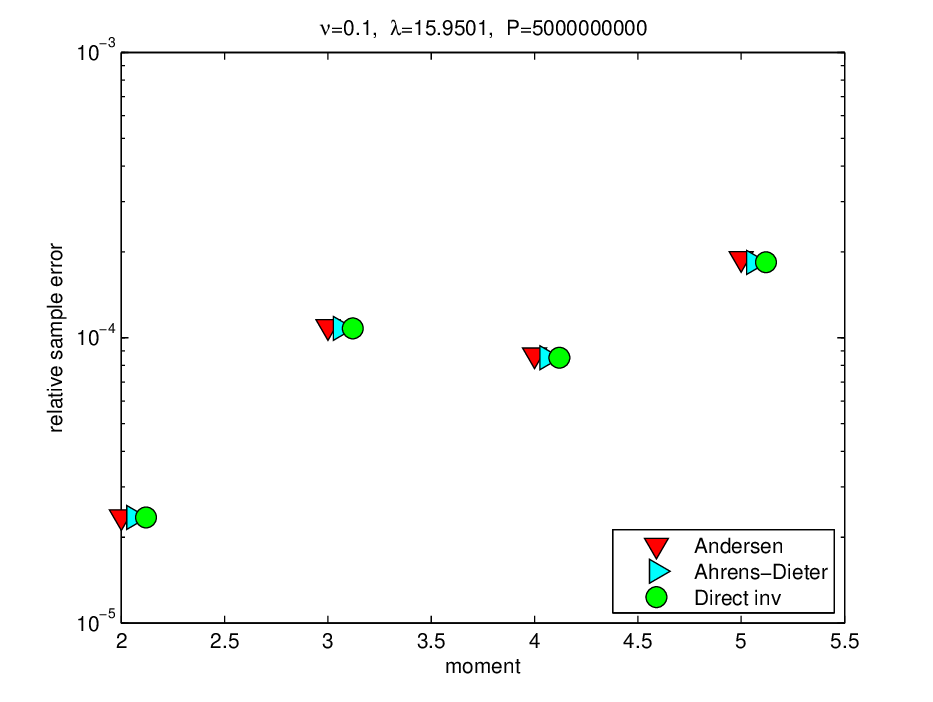}
  \end{center}
\caption{The left panels show the relative error in the first through tenth moments of the 
non-central chi-square sampling methods shown, for the case $\nu=0.1$,
$\lambda=15.9501$ (corresponding to the upper right panel of Figure~\ref{fig:momentserr}).
In the top left panel we used $5\times 10^8$ samples, while in the 
lower left panel we used $5\times 10^9$ samples. The panels on the right show 
the corresponding relative sample errors for the second through fifth moments.}
\label{fig:momentserrlargesample}
\end{figure} 

\begin{figure} 
  \begin{center}
  \includegraphics[width=6.0cm,height=5.0cm]{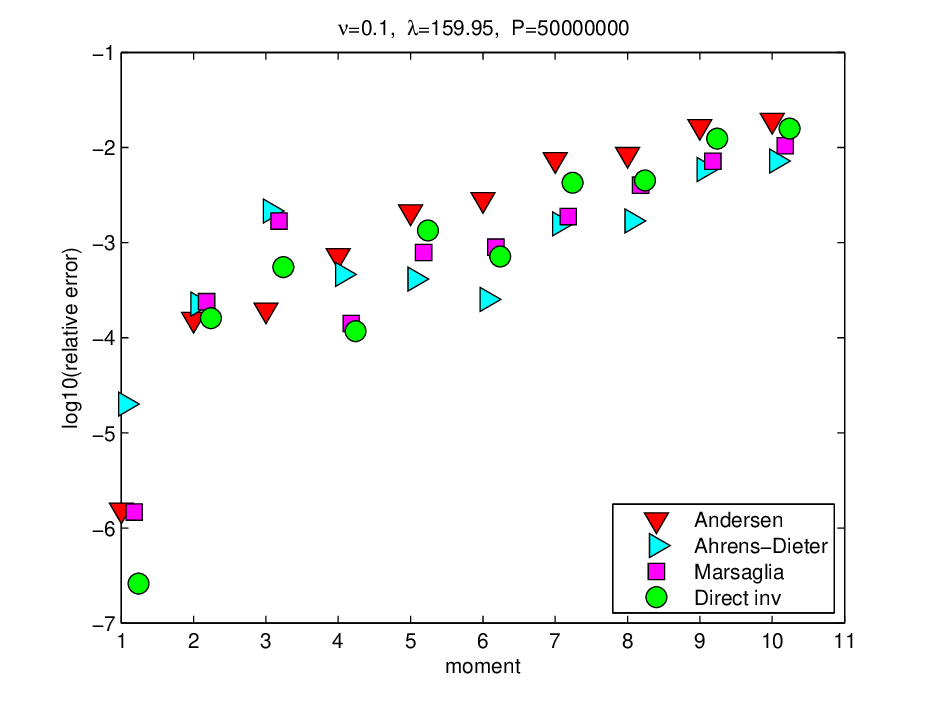}
  \includegraphics[width=6.0cm,height=5.0cm]{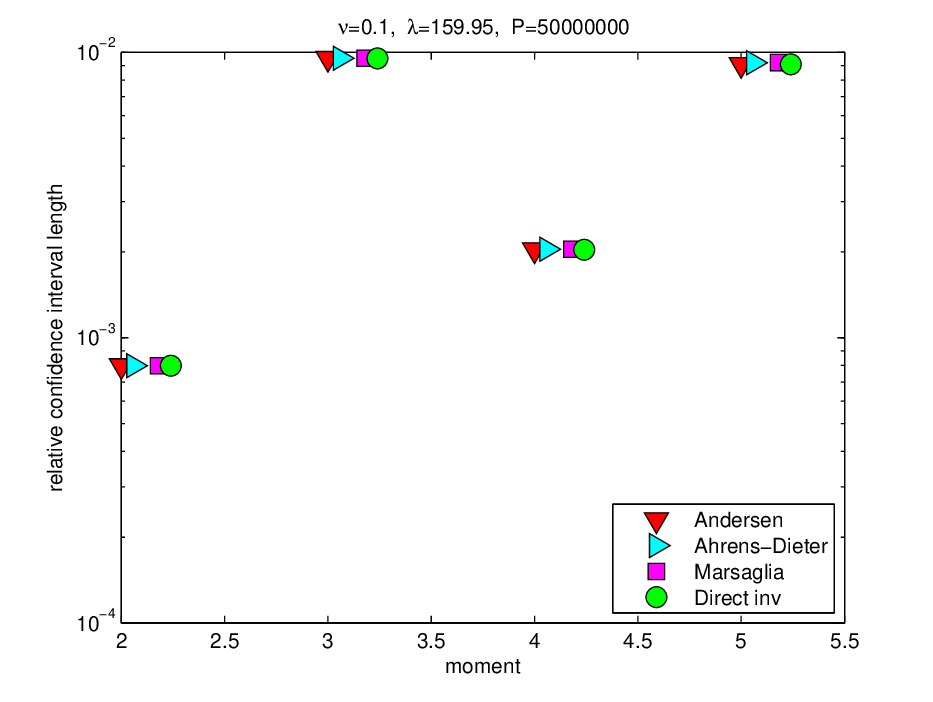}
  \end{center}
\caption{The left panels show the relative error in the first through tenth moments of the 
non-central chi-square sampling methods shown, for the case $\nu=0.1$,
$\lambda=159.95$. We used $5\times 10^7$ samples. The panel on the right shows 
the corresponding relative sample errors for the second through fifth moments.}
\label{fig:momentserrlargenoncentrality}
\end{figure} 

We illustrate in Figures~\ref{fig:icdferr}, 
\ref{fig:momentserr}, \ref{fig:momentserrlargesample} and
\ref{fig:momentserrlargenoncentrality} the performance in
terms of accuracy of our direct inversion and generalized Marsaglia
methods for a representative set of different parameter values for 
$\nu$ and $\lambda$.
In  Figure~\ref{fig:icdferr} the absolute error in the inverse distribution function 
for the direct inversion method  compared with the leading approximation method 
of Andersen (more on this presently) is shown in the case of the 
central chi-square distribution with $\nu=0.02$ degrees of freedom. 
In Figure~\ref{fig:momentserr} the relative errors in the first ten sample central 
moments are displayed for simulating the non-central chi-square distribution 
using our methods. As comparison methods we chose Andersen's approximation method,
whose underlying parameters are fixed to match the first two central moments, and
the exact acceptance-rejection methods of Ahrens--Dieter and Marsaglia--Tsang.
The six panels shown therein correspond to the three values of the degrees
of freedom $\nu=0.1,0.01,0.001$ (top to bottom), and the
values for the non-centrality $\lambda=0.11517,15.9501,0.15505, 15.995,0.1595,15.9995$
(left to right, then top to bottom). The number of samples used
in each case was $5\times 10^7$. We observe that the 
relative errors in all three exact acceptance-rejection methods as 
well as the direct inversion method achieve the same high accuracy in
all ten moments  across all the parameter values. Indeed their accuracy
is essentially limited by the Monte Carlo error which scales as the
reciprocal of the square root of the sample size. Indeed to
confirm this, we see in Figure~\ref{fig:momentserrlargesample}
how their relative errors and relative sample errors 
decrease when the sample size is increased by a factor of 
$10$ and then $10^2$. Any variation in the relative errors
between these four methods across all the plots in Figures~\ref{fig:momentserr}
and \ref{fig:momentserrlargesample} are within the corresponding 
sample errors. For the Andersen approximation method, we note
that in the top left panel and all three right-hand panels in
Figure~\ref{fig:momentserr}, that the first two
moments are indeed matched, while the accuracy in all
the other moments is larger by several orders of magnitude.
Further we observe in Figure~\ref{fig:momentserrlargesample}
that the error in this approximation is invariant to increasing the sample size,
thus exhibiting the bias in this method. 

However we note in the two lower left panels in Figure~\ref{fig:momentserr}
that as the number of degrees of freedom $\nu$ is decreased for small non-centrality, 
the performance of Andersen's approximation improves and indeed matches 
that of the other methods in the case $\nu=0.001$ and $\lambda=0.1595$.
This can be heuristically explained as follows. For small  non-centrality, the
Andersen approximation uses a weighted density function approximation
given by $p\,\delta(0)+(1-p)\,\beta\exp(-\beta x)$,  where the 
parameter $p=(s^2-m^2)/(s^2+m^2)$ characterizes the distribution 
between a mass density at the origin and a decaying exponential approximation
scaled by the parameter $\beta=(1-p)/m$. Here $m$ and $s^2$ represent 
the sample mean and variance, respectively, and this weighted approximation
is invoked when $s^2/m^2>1.5$.  In Figure~\ref{fig:momentserr}
to decrease $\nu$, we decreased $\kappa$; and note that we set
$V(0)=\theta$, $\varepsilon=1$ and $h=1$. A straightforward calculation
using the explicit values for $m$ and $s^2$ given in Andersen~\cite{An},
reveals that $s^2/m^2\sim1/2\theta$ and thus $p\sim(1-2\theta)/(1+2\theta)$ 
for $\kappa\ll 1$. For small $\theta$, we see $p\to1^-$. Note  that $\theta=0.04$
in Figure~\ref{fig:momentserr}. Since for small non-centrality the
non-central chi-square distribution approaches the central chi-square
distribution, and then for small degrees of freedom the central-square
distribution shifts to concentrate high probability of occurrence to
the origin where there is an integrable singularity, that the Andersen
approximation shifts more weight to the point mass at the origin
in this limit naturally shadows this phenomenon. Thus  
in this limit we might expect the Andersen approximation 
to perform better, as indicated in Figure~\ref{fig:momentserr}.

The opposite extreme of the parameter space to this last case
is that of very large non-centrality. We show in 
Figure~\ref{fig:momentserrlargenoncentrality} the 
the relative errors in the first ten sample central moments
for the case when $\nu=0.1$ and $\lambda=159.95$. 
We see that all the methods perform roughly equally
well in terms of accuracy across all the moments
(we have not plotted the Marsaglia--Tsang case as
this is in practice for the parameter regime $\nu <1$ 
of interest here slightly slower than the Ahrens--Dieter
method as we discuss below, though this may not be the case 
for parameter values $\nu >1$). In particular the Andersen
method performs equally well. This is not too surprising.
It reflects the fact that for large centrality 
the Andersen method utilizes an observation by Patnaik~\cite{Patnaik} 
that a very effective approximation to the non-central 
chi-square distribution function is generated by the square 
of a matched normal random variable.

%We note that
%the error in the higher moments for this method 
%are of order 1. Among  the other methods, the relative errors 
%share the same pattern. Since the methods 
%of Marsaglia and Tsang and of Ahrens and Dieter are exact,   
%the observed errors are simply of statistical nature.  
%This observation is supported by the values displayed 
%in Figure~\ref{fig:momentserrlargesample}, 
%the relative errors and how the sample errors 
%decrease when the sample size is increased by the factor 10. Note that
%the relative error in the Andersen approximation method is unchanged, thus
%exhibiting the bias in this method. 
 
The improved accuracy requires a higher computational effort. 
In Table~\ref{table:times} we list the CPU times for each method 
relative to the CPU times of the Andersen method for each of
the cases considered in Figures~\ref{fig:momentserr} and 
\ref{fig:momentserrlargenoncentrality}. For all the cases
considered in Figure~\ref{fig:momentserr} the generalized Marsaglia
and direct inversion methods require on average 
1.41 and 1.46 times respectively more effort. 
The methods of Marsaglia--Tsang and Ahrens--Dieter 
are on average 2.8 and 2.55 times respectively slower
(and thus hereafter we use the method of Ahrens--Dieter
in preference for comparison). For the case in 
Figure~\ref{fig:momentserrlargenoncentrality} the
generalized Marsaglia, direct inversion and Ahrens--Dieter 
methods require 5.25, 5.90 and 6.69 times respectively 
more effort. This is because for high non-centrality
more effort is required to sum the larger number (a Poisson
random variable with mean $\lambda/2$)
of exponential random variables that are used to simulate
the non-central component $\chi^2_0(\lambda)$ of the 
non-central $\chi_\nu^2(\lambda)$ random variable.

To be exhausting comprehensive, we also
calculated  relative CPU times when $\nu=0.777$ and
$\lambda=15.6164$. This parameter value belongs to the computationally 
most expensive parameter cases for the direct inversion method. 
This due to the unfavourable decomposition of $\nu=0.777$ in our fractional
basis $1/5,\, 1/10,\, 1/20,\, 1/50,\, \ldots$, and the requirement
to add powers of multiple generalized Gaussian variables, 
see Section~\ref{subsec:viaggdi}. 
This could be alleviated by using a finer decomposition.  
As expected the Ahrens--Dieter method performance was
unaffected (2.6803) while the direct inversion method
was slower (3.1916).  We emphasize that all our
calculations were performed using compiled Matlab code.
We endeavoured to optimize the performance of all the 
methods, for example using the Beasley--Springer--Moro methods 
for the Gaussian inversion required in the method of Andersen,
and so forth. 

\begin{table}
\begin{center}
\begin{tabular}{lcccc}\hline
$(\nu,\lambda)$ $\phantom{\hat{\Big|}}$& Marsaglia--Tsang & Ahrens--Dieter & Gen. Marsaglia & Direct inv. \\\hline
$(0.1,0.11517)$   & 2.7855  &  2.5689   & 1.4479   & 1.3718 \\  
$(0.1,15.9501)$   & 2.8185  & 2.5683   & 1.6986   & 1.6135 \\
$(0.01,0.1595)$   & 2.7953  & 2.5665   & 1.2207   & 1.4044 \\
$(0.01,15.9995)$  & 2.8137  &  2.5554  &  1.4813  &  1.5990 \\
$(0.001,0.1595)$  & 2.8300  &  2.5451  &  1.1640  &  1.3635\\
$(0.001,15.9995)$ & 2.7525  &   2.5326 &   1.3472 &   1.4918 \\
%$(0.777,15.6164)$ & 2.9133  &  2.6803  & $\cdots$ & 3.1916 \\
$(0.1,159.95)$    &$\cdots$&   6.6921 &   5.2517 &   5.9027 \\\hline
\end{tabular}
\end{center}
\caption{CPU times relative to the Andersen method to compute the 
moments in Figures~\ref{fig:momentserr} and \ref{fig:momentserrlargenoncentrality} 
for the methods and parameters shown.}
\label{table:times}
\end{table}

\subsection{The Heston Model}\label{subsec:Heston}
The CIR process is a main ingredient in the \emph{Heston model} (Heston~\cite{H}).
The Heston  model  is a two-factor model,
in which one component $S$ describes the evolution of a financial variable
such as a stock index or exchange rate,
and the second component $V$ is a CIR process that
describes the stochastic variance of its returns.
It is given by 
\begin{align*}
\mathrm{d}S_t=&\;\mu S_t\,\mathrm{d}t
+\sqrt{V_t}\,S_t\,\bigl(\rho\,\mathrm{d}W^1_t
+\sqrt{1-\rho^2}\,\mathrm{d}W^2_t\bigr),\\
\mathrm{d}V_t=&\;\kappa(\theta-V_t)\,\mathrm{d}t
+\varepsilon\sqrt{V_t}\,\mathrm{d}W^1_t,
\end{align*} 
where $W^1_t$ and $W^2_t$ are independent scalar 
Wiener processes. The parameters $\mu$, $\kappa$, $\theta$
and $\varepsilon$ are all positive and $\rho\in(-1,1)$.
In the context of option pricing, a pricing measure must be
specified. We assume here that the dynamics of $S$ and $V$ as specified
above are given under the pricing measure. For a discussion 
and derivation of various equivalent martingale measures in the Heston
model see for example Hobson~\cite{Ho}.
As noted above, the variance $V$ is non-negative, and the 
stock price $S$, as a pure exponential process,
is positive. Without loss of generality we suppose $\mu=0$.

To estimate the asset price we follow the lead of 
Broadie and Kaya~\cite{BK} and Andersen~\cite{An}
(also see Willard~\cite{Willard} and Romano and Touzi~\cite{RT}).
In the following proposition we assume we have simulated
$V_{t_{n+1}}$ from $V_{t_n}$ exactly---for example using the 
non-central chi-square simulation scheme based on the 
generalized Marsaglia approach.
\begin{proposition}\label{prop:priceprocesssim}
Across the time step $[t_n,t_{n+1}]$, assume 
$V_{t_n}$ and $V_{t_{n+1}}$ are given. Then set
$K_1=h(\kappa\rho/\eps-1/2)/2-\rho/\eps$,
$K_2=h(\kappa\rho/\eps-1/2)/2+\rho/\eps$,
$K_3=h(1-\rho^2)/2$, $s=K_2+K_3/2$,
$\hat s=s\cdot\exp(-\kappa h)/\eta(h)$ and 
for $\hat s<1/2$,
\begin{equation*}
K_0^\ast=-\frac{\lambda\hat s}{1-2\hat s}
+(\nu/2)\cdot\ln(1-2\hat s)-(K_1+K_3/2)V_{t_n}.
\end{equation*}
Then the approximate price process computed as follows
is a martingale:
\begin{equation*}
S_{t_{n+1}}=S_{t_n}\exp\Bigl(K_0^\ast+K_1V_{t_n}+K_2V_{t_{n+1}}+\sqrt{K_3(V_{t_n}+V_{t_{n+1}})}\cdot Z\Bigr),
\end{equation*}
where $Z\sim\mathrm{N}(0,1)$. 
\end{proposition}
\begin{proof}
With $V_{t_n}$ and $V_{t_{n+1}}$ as given, conditioned on the time 
integrated variance across $[t_n,t_{n+1}]$,
we know that $\ln S_{t_{n+1}}-\ln S_{t_n}$ is Normally distributed.  
As suggested by Andersen, we approximate the time integrated variance 
across $[t_n,t_{n+1}]$ by the trapezoidal rule. Exponentiating
we arrive at the scheme from Andersen~\cite[p.~21]{An}:
\begin{equation*}
S_{t_{n+1}}=S_{t_n}\exp\Bigl(K_0+K_1V_{t_n}+K_2V_{t_{n+1}}+\sqrt{K_3(V_{t_n}+V_{t_{n+1}})}\cdot Z\Bigr),
\end{equation*}
where $Z\sim\mathrm{N}(0,1)$ and $K_0=-h\rho\kappa\theta/\eps$.
Then as suggested in Proposition~7 of Andersen~\cite[p.~21]{An},
if we set $M\coloneqq\mathbb E\bigl[\exp(sV_{t_{n+1}})|V_{t_n}\bigr]$
and $K_0^\ast\coloneqq-\ln M-(K_1+K_3/2)V_{t_n}$,
and replace $K_0$ by $K_0^\ast$ in the scheme for $S_{t_{n+1}}$
above, then $\mathbb E[S_{t_{n+1}}|S_{t_n}]=S_{t_n}$.
Hence our task is to compute $M$. Since we simulate $V_{t_{n+1}}$
exactly we know $M=\mathbb E\bigl[\exp(\hat s\cdot z)|V_{t_n}\bigr]$,
where $z\sim\chi_\nu^2(\lambda)$, with $\nu$ and $\lambda$ defined 
for the Heston model. Hence provided $\hat s<1/2$ we have 
$M=\exp\bigl(\lambda\hat s/(1-2\hat s)\bigr)/(1-2\hat s)^{\nu/2}$,
giving the result.
\qed
\end{proof}
\begin{remark}
The requirement $\hat s<1/2$ translates to a mild restriction
on the stepsize $h$, which in practice is not a problem 
(see Andersen~\cite[p.~24]{An}).
\end{remark}
%
%\begin{remark}$\ast$
%In practice if we simulate $V_{t_{n+1}}$ from $V_{t_n}$ using the 
%non-central chi-square simulation method based on direct inversion,
%since the error is so small, the scheme in 
%Proposition~\ref{prop:priceprocesssim} works perfectly well.
%\end{remark}

We test all the methods we have considered, Andersen,
Ahrens--Dieter, generalized Marsaglia and direct inversion
for pricing five practical and challenging options. 
We use Andersen's test cases I--III for pricing long-dated 
European call options (maturing at time $T$). 
Andersen describes case I as typical for FX markets, case II as
typical for long-dated interest rate markets and case III as possible
in equity option markets. We also considered Smith's test case for an Asian option 
(see Smith~\cite{Smith} and Haastrecht and Pelsser~\cite{HP})
and Lord, Koekkoek and Van Dijk's test case for a digital double no touch
barrier option. The parameter values for all five cases are shown 
in Table~\ref{table:cases}. Note that in case III, we have assumed the
risk-free rate of interest $r=0.05$, as in Haastrecht and Pelsser~\cite{HP}. 
Let the exact option price at maturity be $C$. The error of the approximation 
is $E=C-\hat C$, where $\hat{C}$ is the sample average of the simulated option payout
at maturity. In our examples, we use a sample size of $10^6$ (except for
the barrier option case). The performance of the method of
Haastrecht and Pelsser~\cite{HP} is similar to Andersen's; 
the reader interested in the actual comparisons is referred to 
their paper. In Table~\ref{table:Andcomptimings} we show,
for the test cases I--III, the relative CPU times to required to compute 
the option prices compared to Andersen's method. The errors 
at three different strikes $100,140,60$, which are dominated by
the trapezoidal rule approximation in the price process, are all comparable to
Andersen's method and so we omit them. We did not implement
any postprocessing such as variance reduction here. We see from
Table~\ref{table:Andcomptimings} that for these plain
vanilla option cases, Andersen's method is in fact the slowest.

We show in Table~\ref{table:asianbias} the simulation results for 
the Asian option with yearly fixings, with very similar conclusions
in terms of accuracy. The generalized Marsaglia and direct inverse
methods are now almost two times slower than Andersen's method
in this case due to the slightly unfavourable form of the degrees
of freedom $\nu=0.904$ (in these two cases we rounded
off the exact degrees of freedom $\nu=0.9035$).  However the 
accuracy they deliver for the variance process far outweighs
their relative speed.

We also apply the four methods to pricing
a digital double no touch barrier option---such an option pays one unit 
of currency if neither barrier is touched and zero if one is.
We monitor at each timestep to determine if either of the barriers had 
been crossed. Indeed, we show in Table~\ref{table:DDNTbias} 
our simulation results. In terms of accuracy for the stepsizes shown,
all the methods perform equally well. In terms of CPU time,
all the methods are faster or roughly the same speed as 
Andersen's method, though  the form of the 
number of degrees of freedom in this case $\nu=0.08$ favours 
the generalized Marsaglia and direct inversion methods.
Note that small timesteps are considered in this test case. 
This means that the non-centrality parameter $\lambda$ can
be large for some time intervals; in which case we use 
Algorithm~\ref{alg:Poisson2} to generate chi-square samples.
%
%\begin{remark}$\ast$
%Insert timings for non-favourable test case for 250 and 500 timesteps.
%\end{remark}
%

\begin{table}
\begin{center}
\begin{tabular}{cccccc}\hline
$\phantom{\hat{\Big|}}$ Parameters & Case I & Case II & Case III & Case Asian & Case DDNT\\\hline
$\eps$         & 1.0   & 0.9  & 1.0  & 0.5196  & 1.0  \\
$\kappa$    & 0.5   & 0.3  & 1.0  & 1.0407  & 0.5  \\
$\rho$         &-0.9   &-0.5  & -0.3 & -0.6747 & 0.0  \\
$T$              & 10    & 15   & 5    & 4       & 1    \\
$\theta$     & 0.04  & 0.04 & 0.09 & 0.0586  & 0.04 \\
$S(0)$         & 100   & 100  & 100  & 100     & 100  \\
$V(0)$         & 0.04  & 0.04 & 0.09 & 0.0194  & 0.04 \\
$r$               & 0.0   & 0.0  & 0.05 & 0.0     & 0.0  \\\hline
\end{tabular}
\end{center}
\caption{Cases I---III are from Andersen, while Case Asian is from 
Smith and Case DDNT (digital double no touch barrier option) 
is from Lord, Koekkoek and Van Dijk. Here $r$ is the risk-free rate of interest.}
\label{table:cases}
\end{table}

\begin{table}
\begin{center}
\begin{tabular}{lccc}\hline
Case $\phantom{\hat{\Big|}}$& Ahrens--Dieter & Marsaglia & Direct inv.\\\hline
I  & 0.62 &  0.27 &  0.26 \\  
II & 0.59 &  0.80 &  0.55 \\
III& 0.57 &  1.04 & 0.39\\\hline
\end{tabular}
\end{center}
\caption{Cases I---III from Andersen: we show the relative CPU times
(to the method of Andersen) for pricing the options concerned
using $10^6$ paths. The values shown are averaged across 
stepsizes $1/4$, $1/8$, $1/16$ and $1/32$---there was hardly
any variation for different stepsizes.}
\label{table:Andcomptimings}
\end{table}

\begin{table}
\begin{center}
\begin{tabular}{lcccc}\hline
Stepsize $\phantom{\hat{\Big|}}$& Andersen & Ahrens--Dieter &Marsaglia  & Direct inv. \\\hline
1/4  & [-0.0113,1] & [-0.0090,1.41]& [-0.0287,2.20] & [0.0719,2.19]\\  
1/8  & [-0.0166,1] & [-0.0151,1.30]& [-0.0073,2.02] & [0.0519,1.99]\\
1/16& [-0.0175,1]& [-0.0082,1.26]& [0.0240,1.95]   & [0.0628,1.90]\\
1/32& [-0.0287,1]& [-0.0376,1.24]& [-0.0393,1.90]  & [0.0604,1.86]\\\hline
\end{tabular}
\end{center}
\caption{Case Asian from Smith: Estimated error using $10^6$ paths for 
at the money Asian option (strike $100$) with 
yearly fixings. In each case the two-tuple shown is the estimated error
and relative CPU time required to compute the option price. 
In all cases the sample standard deviation was $0.014$.
All estimated errors are not statistically significant at the level 
of three sample standard deviations.
%The exact option price from Smith is 9.7199.
}
\label{table:asianbias}
\end{table}

\begin{table}
\begin{center}
\begin{tabular}{lcccc}\hline
$h$ $\phantom{\hat{\Big|}}$& Andersen & Ahrens--Dieter & Marsaglia & Direct inv.  \\\hline
1/250  & [0.5266,2.00,1]& [0.5300,0.63,0.76]& [0.5238,2.00,0.81]& [0.5300,2.00,0.96] \\
1/500  & [0.5205,1.00,1]& [0.5208,0.32,0.77]& [0.5191,1.00,0.83]& [0.5194,1.00,0.99] \\
1/1000& [0.5154,0.50,1]& [0.5150,0.16,0.78]& [0.5147,0.50,0.87]& [0.5148,0.50,1.03]\\
1/2000& [0.5111,0.25,1]& [0.5105,0.08,0.86]& [0.5109,0.25,0.87]& [0.5108,0.25,1.02]\\\hline
\end{tabular}
\end{center}
\caption{Case DDNT from Lord, Koekkoek and Van Dijk: 
Estimated option price for the digital double no touch
barrier option, using $1/h^2$ paths where $h$ is the stepsize. 
The barriers are 110 and 90. In each case the triple shown is the estimated price,
the sample standard deviation (inflated by $10^3$) and relative CPU time 
required to compute the option price. 
%The exact option price from Lord, Koekkoek and Van Dijk is 0.5011.
}
\label{table:DDNTbias}
\end{table}

\begin{remark}
Note that for cases I--III we could improve the efficiency of 
the algorithm we have implemented as follows (and with mild 
modification to the Asian option with yearly fixings as well).
We decompose $\int{V_\tau}\,\mathrm{d}\tau$ on $[0,T]$ 
into subintervals $[t_n,t_{n+1}]$, use a simple quadrature to approximate
$\int{V_\tau}\,\mathrm{d}\tau$ on these subintervals much like Andersen, 
and simulate the transition densities required using the generalized Marsaglia
method. We then only exponentiate at the final time $T$ to generate an 
approximation for $S_T$ (since we do not compute the price process at
each timestep, this will be more efficient).
However, one advantage of the approach we have taken in this
paper for simulating the price process 
based on the method proposed by Andersen, is that it is more flexible.
For example, it allows us to consider pricing path-dependent options.
\end{remark}
\begin{remark}
Glasserman and Kim~\cite{GK} have recently introduced
a novel method for simulating the time integrated variance process in
the Heston model (also see Chan and Joshi~\cite{CJ}). 
As we can see from our analysis above, to compute the
price process at the end-time $T$, we in essence need to sample from
the distribution for $\int{V_\tau}\,\mathrm{d}\tau$ on the interval $[0,T]$.
The transition density for this integral process over the whole
interval $[0,T]$, given $V_0$ and $V_T$,
is well known and given in Pitman and Yor~\cite{PY}. 
Its Laplace or Fourier transform has a closed form.
Broadie and Kaya~\cite{BK} use Fourier inversion techniques to
sample from this transition density for $\int{V_\tau}\,\mathrm{d}\tau$.
Glasserman and Kim instead separate the Laplace transform of this
transition density into constituent factors, each of which can
be interpreted as the Laplace transforms of probability densities,
samples of which can be generated by series of particular
gamma random variables. The advantage of this method is that
$\int{V_\tau}\,\mathrm{d}\tau$ is simulated
directly on the interval $[0,T]$. Glasserman and Kim have demonstrated that
this is an efficient alternative to the quadrature approximation
of the time integrated variance suggested by Andersen,
if one is interested in the pricing non-path-dependent derivatives, 
which does not require the simulation of any intermediate
values of the asset process $S$. They also note that
when pricing path-dependent options, quadrature approximation
of the time integral of the variance process will be more efficient
(see end of their Section~5).
\end{remark}

\section{Concluding remarks}\label{sec:conclu}
We have introduced two new methods for sampling the 
CIR non-central chi-square process. The first is
the generalized Marsaglia method which is an 
exact acceptance-rejection method. The second
is a direct inversion method based on the 
Beasley--Springer--Moro method which delivers
very high accuracy, and which in principle can be
extended to machine accuracy (double precision). This method
has the advantage of being 
amenable to implementation and simulation
using  quasi-Monte Carlo sequences as well as for
sensitivity analysis. % (one of our next objectives). 
Both methods are easy to implement and
flexible as the CIR process, which serves as a fundamental building
block in many financial models, 
can be immediately simulated for
any value of degrees of freedom. We illustrated their accuracy
and their efficiency 
for an extensive range of parameter values.
The efficiency performance of both methods  are similar and compare
well with other leading chi-square sampling methods.
%look good. Simulation CPU times depend on the 
%degrees of freedom as well as the non-centrality 
%if this can be large. At the very worst, in a specifically
%concocted artificial examples, both methods 
%can be ten times slower. However generically 
%they perform the same as or better than
%the leading method of Andersen whilst delivering
%many more orders of magnitude of accuracy.
We illustrated the use of our new methods for the
simulation of the Heston model.
In terms of simulating the Heston model the
accuracy delivered for the variance process
is somewhat overridden by the error associated with
the trapezoidal rule approximation used in the 
price process simulation. We expect that if
the new more accurate approximation method 
of Glasserman and Kim for the integrated 
variance process is used instead, the accuracy
available for the variance process will become
more prevalent. Lastly, another additional
direction of interest would be to consider how
to optimize both our methods for use in
general processing units (GPUs); see for example 
Giles~\cite{Giles} who considers an efficient
approximation of the inverse error function for
GPU execution.

\begin{acknowledgements}
We would like to thank two previous referees.
One for encouraging us to perform some numerical 
comparisons and also for pointing out the limit involving the 
Euler--Mascheroni constant, and another for bringing the 
Marsaglia--Tsang gamma distribution sampling method 
to our attention. We also like to thank Gavin Gibson,  
Mark Owen and Karel in 't Hout for stimulating discussions. 
A part of this paper was completed while AW was visiting 
the Fields Institute, Toronto, in Spring 2010.
Their support during the visit is gratefully acknowledged.
\end{acknowledgements}

\newpage

\appendix

\section{Generalized Gaussian direct inversion algorithm}\label{app:di}
Here we have assumed $(3,4)$ Pad\'e approximants in both the central and middle 
regions and a degree $10$ Chebychev approximant in the tail region.
Adapting the algorithm to other degree approximants is straightforward.
We assume $q$ is given and the parameters $\gamma_q$, $C_q$, $\Phi_\pm$, $\eta_*$, 
$k_1$ and $k_2$ defined in Section~\ref{subsec:di} have been calculated as well. 
In the algorithm below these parameters are: \texttt{gamma\_q}, \texttt{C\_q},
\texttt{Phi\_minus}, \texttt{Phi\_plus}, \texttt{eta\_star}, \texttt{k\_1} and
\texttt{k\_2}. Further the coefficients $a_0,a_1,a_2,a_3,c_0,c_1,c_2,c_3$ are
stored as the vectors \texttt{a} and \texttt{c} with index starting at \texttt{1},
so $a_0$ corresponds to \texttt{a(1)}, $a_1$ to \texttt{a(2)}, etc., whereas
the coefficients $b_1,b_2,b_3,b_4,d_1,d_2,d_3,d_4$ are
stored as the vectors \texttt{b} and \texttt{d} with exact indexing correspondence.
The coefficients $\hat c_0,\ldots,\hat c_{10}$ are stored as the vector \texttt{c\_\,hat}
with index running from \texttt{1} to \texttt{11}.
The input \texttt{U} is a $\text{U}(0,1)$ uniform random variable and the output \texttt{X} is
a generalized Gaussian random variable.
\lstset{language=Matlab,basicstyle=\ttfamily}
\begin{lstlisting}[frame=topline,caption={Generalized Gaussian direct inversion},label=alg:ggdil]
Y=(U-0.5)/gamma_q;
if (abs(Y)<(Phi_minus-0.5)/gamma_q)
   R=Y^q;
   X=Y*(((a(4)*R+a(3))*R+a(2))*R+a(1)) ...
       /((((b(4)*R+b(3))*R+b(2))*R+b(1))*R+1.0);
else
   R=1-U;
   if (Y<0)
       R=U;
   end
   if (abs(U-0.5)<Phi_plus-0.5)
      R=-log(R)-eta_star;
      X=(((c(4)*R+c(3))*R+c(2))*R+c(0)) ...
        /((((d(4)*R+d(3))*R+d(2))*R+d(1))*R+1.0);
   else
      R=k_1*log(-log(R/C_q))+k_2;
      D2=0;
      D1=0;
      for j=10:-1:1
        D=2*R*D1-D2+c_hat(j+1);
        D2=D1;
        D1=D;
      end
      X=R*D1-D2+0.5*c_hat(1);
   end
   if (Y<0)
      X=-X;
   end
end
\end{lstlisting}

\newpage

\section{Direct inversion coefficients}\label{app:coeffs}

\begin{center}
\begin{tabular}{ccccc}\hline\hline
\multicolumn{5}{c}{$q=10$ $\phantom{\hat{\Big|}}$}\\\hline
$\phantom{\hat{\Big|}}$ 
n & $a_n$               & $c_n$               &$n$ & $\hat c_n$       \\\hline
0 & 0.999999999999651   &  1.060540481693800 & 0  & 2.622284617034058 \\
1 & -1.429881128897603  &  0.796155091482938 & 1  & 0.1449805130767122 \\
2 & 0.601262815177118   &  0.206235219404016 & 2  & 0.003259370325482870 \\
3 & -0.068206095200774  &  0.020226513592948 & 3  & -0.0003397434921419157  \\
4 &                     &  0.000479843137311 & 4  & 0.00007014928432054771   \\
  &                     &                    & 5  & -0.000003586305447050563   \\
  &    $b_n$            &      $d_n$          & 6  & -7.631531772738493$*10^{-7}$  \\
1 &  -1.475335674435254 & 0.683564944492548  & 7  & 1.840112411709724$*10^{-7}$  \\
2 &  0.651548639035629  & 0.161851250036749  & 8  & -1.217436540241387$*10^{-8}$ \\
3 &  -0.081616351333977 & 0.014123257065970  & 9  & -2.007039183742053$*10^{-9}$ \\
4 &  0.000391957158842  & 0.000270655354670  & 10 & 5.694689247537491$*10^{-10}$  \\
  &                     &                    &    &                    \\
  &       $\Phi_-$     &  $\Phi_+$            &    & $\eta_*$          \\
  & 0.954178994865017  & 0.998325461835062   &    & 4.254756463685820  \\
  &       $k_1$        &  $k_2$               &    &                   \\
  & 1.015803736413048  & -2.256872281479897  &    &                    \\
\hline
\multicolumn{5}{c}{$q=100$ $\phantom{\hat{\Big|}}$}\\\hline
$\phantom{\hat{\Big|}}$ 
n & $a_n$               & $c_n$               &$n$ & $\hat c_n$                 \\\hline
0 & 0.999999999999675  &  1.006854352727258  & 0  &  2.053881658435666       \\
1 & -1.582783912975250 &  0.731099829867281  & 1  &  0.01034073560906051   \\
2 & 0.745662878312873  &  0.195140887172246  & 2  &  -0.00002227561625609034   \\
3 & -0.097011209317889 &  0.021377534151187  & 3  &  -0.00002539481838124709  \\
4 &                    &  0.000693785524959  & 4  &   0.000005279864625853940 \\
  &                    &                     & 5  &  -3.805077742014832$*10^{-7}$ \\
  &       $b_n$        &       $d_n$          & 6  &  -3.298923929556226$*10^{-8}$\\
1 & -1.587734408086399 & 0.720091560042297   & 7  &  1.114785567355588$*10^{-8}$ \\
2 & 0.751669594595948  & 0.190684456942316   & 8  &  -9.757310866073955$*10^{-10}$\\
3 & -0.098779495517049 & 0.020677157072949   & 9  &  -6.769928792613323$*10^{-11}$ \\
4 & 0.000055151948082  & 0.000659853425082   & 10 & 2.803447160471445$*10^{-11}$\\
5 &                    & -0.000000097154009  &    &                            \\
  &                    &                     &    &                             \\
  &       $\Phi_-$     &  $\Phi_+$            &    & $\eta_*$                   \\
  &  0.996245001605534 & 0.999888263643581   &    &  6.788371878124332          \\
  &       $k_1$        &  $k_2$               &    &                             \\
  & 1.130241473667677  &  -2.510876893558557 &    &                             \\
\hline
\multicolumn{5}{c}{$q=1000$ $\phantom{\hat{\Big|}}$}\\\hline
$\phantom{\hat{\Big|}}$ 
n & $a_n$               & $c_n$              &$n$ & $\hat c_n$                 \\\hline
0 & 0.999999999996602  & 1.000692386269727  &  0 &  2.005202593715361\\
1 & -2.214484997909744 & 0.641334743798204  &  1 & 0.0009445439483225688 \\
2 & 1.455225242281931  & 0.146036839714129  &  2 & -0.000003302159950131867 \\
3 & -0.270311067182453 & 0.012706381032885  &  3 & -0.000002214168006581846 \\
  &                    & 0.000254873053744  &  4 &  4.226488352359718$*10^{-7}$\\
  &                    &                    &  5 & -2.701556647692559$*10^{-8}$ \\
  & $b_n$              &  $d_n$              &  6 & -2.660957832417678$*10^{-9}$ \\
1 & -2.214984498880292 & 0.640293981966015  &  7 & 7.589961583952764$*10^{-10}$\\
2 & 1.456144357614843  & 0.145674000937359  &  8 & -5.604691330255176$*10^{-11}$ \\
3 & -0.270724201676657 & 0.012660105958367  &  9 & -5.184165197371945$*10^{-12}$ \\
4 & 0.000022976139411  & 0.000253400757973  & 10 &  1.626439890027763$*10^{-12}$\\
  &                    &                     &    &                             \\
  &       $\Phi_-$     &  $\Phi_+$            &    & $\eta_*$                   \\
  & 0.999632340672519  & 0.999989279922375   &    & 9.115135573141224          \\
  &       $k_1$        &  $k_2$               &    &                             \\
  & 1.211390454015218  & -2.630859238259484  &    &                             \\
\hline\hline
\end{tabular}
\end{center}

\newpage

\begin{center}
\begin{tabular}{ccccc}\hline\hline
\multicolumn{5}{c}{$q=5$ $\phantom{\hat{\Big|}}$}\\\hline
$\phantom{\hat{\Big|}}$ 
n & $a_n$               & $c_n$               &$n$ & $\hat c_n$       \\\hline
0 & 0.999999999999962   & 1.098560543273500  & 0  & 3.446913820849123\\
1 & -1.288113131377250  & 1.076929115611482  & 1  & 0.4032831311503550 \\
2 &  0.481578771462415  & 0.374009830584217  & 2  & 0.02171691866430493 \\
3 & -0.047325498551885  & 0.052979687032815  & 3  & -0.0003693910171288154  \\
4 &                     & 0.002320423236613  & 4  & 0.0001643478336907373  \\
  &                     &                    & 5  &  -9.968595138386470$*10^{-7}$ \\
  &    $b_n$            &      $d_n$          & 6  &  -0.000002678275851276468\\
1 & -1.371446464722253  &  0.826900637356423 & 7  &  4.340978214450863$*10^{-7}$  \\
2 & 0.565562947138128   & 0.243236630017604  & 8  &  -8.359190851308088$*10^{-9}$ \\
3 & -0.067614258837771  & 0.028035297860946  & 9  &  -8.345847144867538$*10^{-9}$\\
4 & 0.000560269685859   & 0.000866824877085  & 10 &  1.501644408119446$*10^{-9}$ \\
5 &                     & -0.000003203247416 &    &                    \\
  &                     &                    &    &                    \\
  &       $\Phi_-$      &  $\Phi_+$           &    & $\eta_*$          \\
  & 0.888435024173769  & 0.994853658080896   &    & 3.327051730489134  \\
  &       $k_1$         &  $k_2$              &    &                   \\
  & 0.9436583821081551  & -2.053011104657458 &    &                    \\
\hline 
\multicolumn{5}{c}{$q=50$ $\phantom{\hat{\Big|}}$}\\\hline
$\phantom{\hat{\Big|}}$ 
n & $a_n$               & $c_n$               &$n$ & $\hat c_n$       \\\hline
0 & 0.999999999999476   & 1.013549868031473  & 0  & 2.109911415053862\\
1 & -1.564458809116706  & 0.667930936229205  & 1  & 0.02170729783674736 \\
2 & 0.727524267692390   & 0.155499481214690  & 2  & 0.000005070046142588313 \\
3 & -0.093190467403288  & 0.013822381509740  & 3  & -0.00005423065778776348  \\
4 &                     & 0.000286136307044  & 4  & 0.00001132985142947280  \\
  &                     &                    & 5  &  -8.168939038096652$*10^{-7}$ \\
  &    $b_n$            &      $d_n$          & 6  & -7.483150884891461$*10^{-8}$ \\
1 &  -1.574262730786144 & 0.646800829595665  & 7  & 2.519301860662225$*10^{-8}$   \\
2 & 0.739293882226217   & 0.147982563681919  & 8  & -2.231412272701264$*10^{-9}$  \\
3 & -0.096615964501066  & 0.012852173023820  & 9  & -1.581249637852744$*10^{-10}$ \\
4 & 0.000106012694153   & 0.000254893106908  & 10 & 6.651970666939820$*10^{-11}$  \\
  &                     &                    &    &                    \\
  &       $\Phi_-$      &  $\Phi_+$           &    & $\eta_*$          \\
  & 0.992313833379312   & 0.999766047505894  &    & 6.068592841104139  \\
  &       $k_1$         &  $k_2$              &    &                   \\
  & 1.104377984691796   & -2.464549291690036 &    &                    \\
\hline
\multicolumn{5}{c}{$q=500$ $\phantom{\hat{\Big|}}$}\\\hline
$\phantom{\hat{\Big|}}$ 
n & $a_n$               & $c_n$               &$n$ & $\hat c_n$       \\\hline
0 & 1.000000000001737   & 1.001383246165305  & 0  & 2.010495391375142\\
1 & -1.391636943669522  & 0.643915677152308  & 1  & 0.001933087453438041  \\
2 & 0.529242181140450   & 0.147299234948298  & 2  & -0.000006861856172820585 \\
3 & -0.043796708347591  & 0.012900083582821  & 3  & -0.000004593728077441438  \\
4 &                     & 0.000260965296708  & 4  & 9.054993169060948$*10^{-7}$  \\
  &                     &                    & 5  &  -6.064260759136934$*10^{-8}$ \\
  &    $b_n$            &      $d_n$          & 6  & -5.622897885014598$*10^{-9}$ \\
1 & -1.392634947413242  & 0.641830834617498  & 7  & 1.709975585729971$*10^{-9}$   \\
2 & 0.530257903665352   & 0.146569705446697  & 8  & -1.346040760729119$*10^{-10}$  \\
3 & -0.044005529297217  & 0.012806476344300  & 9  & -1.115867483028221$*10^{-11}$ \\
4 &                     & 0.000257963030654  & 10 &  3.854801783433290$*10^{-12}$ \\
  &                     &                    &    &                    \\
  &       $\Phi_-$      &  $\Phi_+$           &    & $\eta_*$          \\
  & 0.999262947245193   & 0.999978460704705  &    & 8.419292019151525  \\
  &       $k_1$         &  $k_2$              &    &                   \\
  & 1.186173840297473   & -2.595663671659387 &    &                    \\
\hline\hline
\end{tabular}
\end{center}

\newpage

\begin{center}
\begin{tabular}{ccccc}\hline\hline
\multicolumn{5}{c}{$q=20$ $\phantom{\hat{\Big|}}$}\\\hline
$\phantom{\hat{\Big|}}$ 
n & $a_n$               & $c_n$               &$n$ & $\hat c_n$       \\\hline
0 & 0.999999999999362   & 1.032613218276406  & 0  & 2.288521173202021\\
1 & -1.511386771163247  & 0.713123517283527  & 1  & 0.06094747962391468  \\
2 & 0.676248487105209   & 0.172466963210780  & 2  & 0.0004979897167054818 \\
3 & -0.082692565611279  & 0.015897412011461  & 3  & -0.0001543431236869771  \\
4 &                     & 0.000346782310140  & 4  & 0.00003160495474050310  \\
  &                     &                    & 5  &  -0.000002097939830730025 \\
  &    $b_n$            &      $d_n$          & 6  & -2.545431143318975$*10^{-7}$ \\
1 & -1.535196295102288  & 0.659013658465100  & 7  & 7.644348818405340$*10^{-8}$   \\
2 & 0.703944802836950   & 0.152518516830577  & 8  & -6.332863001383525$*10^{-9}$  \\
3 & -0.090485340265877  & 0.013264787335393  & 9  & -5.895856800809525$*10^{-10}$ \\
4 & 0.000235676079099   & 0.000260023340056  & 10 & 2.181280362294034$*10^{-10}$  \\
  &                     &                    &    &                    \\
  &       $\Phi_-$      &  $\Phi_+$           &    & $\eta_*$          \\
  & 0.979433650152057   & 0.999329809791150  &    & 5.073863838784834  \\
  &       $k_1$         &  $k_2$              &    &                   \\
  & 1.062049352492145   & -2.373877078913848  &    &                    \\
\hline 
\multicolumn{5}{c}{$q=200$ $\phantom{\hat{\Big|}}$}\\\hline
$\phantom{\hat{\Big|}}$ 
n & $a_n$               & $c_n$               &$n$ & $\hat c_n$       \\\hline
0 & 0.999999999999818   & 1.003446599107830   & 0  & 2.026585952893378\\
1 & -1.592059576219168  & 0.646500675620922   & 1  & 0.005000572651842971 \\
2 & 0.754928347929758   & 0.147774949658364   & 2  & -0.00001586472529618898 \\
3 & -0.098984011870256  & 0.012900041834309   & 3  & -0.00001209785441132469  \\
4 &                     & 0.000259966165110   & 4  & 0.000002470674959854400  \\
  &                     &                     & 5  &  -1.739180236341302$*10^{-7}$ \\
  &    $b_n$            &      $d_n$           & 6  & -1.521404951930557$*10^{-8}$ \\
1 & -1.594547138442268  & 0.641281532007261   & 7  & 4.971592453081240$*10^{-9}$   \\
2 & 0.757962823762529   & 0.145949090224436   & 8  & -4.195656728956311$*10^{-10}$  \\
3 & -0.099882437363767  & 0.012666874097162   & 9  & -3.074520478566509$*10^{-11}$ \\
4 & 0.000028130467837   & 0.000252543685829   & 10 & 1.193471454418927$*10^{-11}$  \\
  &                     &                     &    &                    \\
  &       $\Phi_-$      &  $\Phi_+$            &    & $\eta_*$          \\
  & 0.998144331394750   & 0.999945402061219   &    & 7.494926977014854 \\
  &       $k_1$         &  $k_2$               &    &                   \\
  & 1.154337013616336   & -2.549188505020307  &    &                    \\
\hline
\multicolumn{5}{c}{$q=2000$ $\phantom{\hat{\Big|}}$}\\\hline
$\phantom{\hat{\Big|}}$ 
n & $a_n$               & $c_n$               &$n$ & $\hat c_n$       \\\hline
0 & 0.999999999997019   & 1.000346383496690  & 0  & 2.002579775991956\\
1 & -1.455537231516898  & 0.688914570824710  & 1  & 0.0004617142213663357  \\
2 & 0.588358222689532   & 0.174806929040730  & 2  & -0.000001526635305931853 \\
3 & -0.053434409210900  & 0.017594316937225  & 3  & -0.000001066615247993659  \\
4 &                     & 0.000421430584942  & 4  & 1.965467095261218$*10^{-7}$  \\
  &                     &                    & 5  & -1.191943072817269$*10^{-8}$  \\
  &    $b_n$            &      $d_n$          & 6  & -1.255808491556965$*10^{-9}$ \\
1 & -1.455787106562329  & 0.688377683635393  & 7  & 3.345974182790842$*10^{-10}$  \\
2 & 0.588628288218368   & 0.174601235888502  & 8  & -2.298521347563630$*10^{-11}$  \\
3 & -0.053495358024573  & 0.017563762752719  & 9  & -2.387233298590672$*10^{-12}$ \\
4 &                     & 0.000420251890153  & 10 & 6.789837048813090$*10^{-13}$  \\
  &                     &                    &    &                    \\
  &       $\Phi_-$      &  $\Phi_+$           &    & $\eta_*$          \\
  & 0.999816386904579   & 0.999994652315274  &    & 9.809631850391680  \\
  &       $k_1$         &  $k_2$              &    &                   \\
  & 1.238207674406019   & -2.667612981028375 &    &                    \\
\hline\hline
\end{tabular}
\end{center}

\end{document}